\documentclass[sigconf]{acmart}

\renewcommand\footnotetextcopyrightpermission[1]{} 

\usepackage{bussproofs}
\usepackage{amsmath}
\usepackage{xcolor}
\usepackage{mathtools}
\usepackage{graphicx}
\usepackage[skip=0.5ex]{subcaption}
\usepackage{stmaryrd}
\usepackage{xspace}

\AtBeginDocument{%
  \providecommand\BibTeX{{%
    \normalfont B\kern-0.5em{\scshape i\kern-0.25em b}\kern-0.8em\TeX}}}
    
\newif\ifdraft
\draftfalse

\newcommand\dia{\lozenge}
\newcommand\sqand{\sqcap}
\newcommand\sqor{\sqcup}
\newcommand\boxp[1]{[#1]}
\newcommand\diap[1]{\langle #1\rangle}
\newcommand\hp[1]{\llbracket #1\rrbracket}

\newcommand\first[1]{{\sf first}\,#1}
\newcommand\sta[1]{{\sf Sta}\,#1}
\newcommand\tra[1]{{\sf Tra}\,#1}
\newcommand\var[1]{{\sf Var}\,#1}
\newcommand\last[1]{{\sf last}\,#1}
\newcommand\true[1]{{\sf true}\,#1}
\newcommand\false[1]{{\sf false}\,#1}

\newcommand\stdl{STd$\mathcal{L}$\xspace}
\newcommand\dl{d$\mathcal{L}$\xspace}

\begin{document}

\title{A Program Logic to Verify Signal Temporal Logic Specifications of Hybrid Systems: Extended Technical Report}
\author{Hammad Ahmad}
\email{hammada@umich.edu}
\affiliation{%
  \institution{University of Michigan, Ann Arbor}
  \city{Ann Arbor}
  \state{Michigan}
  \postcode{48109}
}

\author{Jean-Baptiste Jeannin}
\email{jeannin@umich.edu}
\affiliation{%
  \institution{University of Michigan, Ann Arbor}
  \city{Ann Arbor}
  \state{Michigan}
  \postcode{48109}
}


\begin{abstract}
Signal temporal logic (STL) was introduced for monitoring temporal properties of continuous-time signals for continuous and hybrid systems. Differential dynamic logic (\dl) was introduced to reason about the end states of a hybrid program. Over the past decade, STL and its variants have significantly gained in popularity in the industry for monitoring purposes, while \dl has gained in popularity for verification of hybrid systems. In this paper, we bridge the gap between the two different logics by introducing signal temporal dynamic logic (\stdl) -- a dynamic logic that reasons about a subset of STL specifications over executions of hybrid systems. Our work demonstrates that STL can be used for deductive verification of hybrid systems. \stdl significantly augments the expressiveness of \dl by allowing reasoning about temporal properties in given time intervals. We provide a semantics and a proof calculus for \stdl, along with a proof of soundness and relative completeness.\footnote{This technical report is an extended version with detailed proofs of the paper ``A Program Logic to Verify Signal Temporal Logic Specifications of Hybrid Systems'' that appeared at HSCC 2021\cite{ahmad2021stdl}.}
\end{abstract}

\maketitle
\pagestyle{plain} 


\section{Introduction} \label{sec:intro}
Recent technological advances have made our transportation, manufacturing and communication facilities safer, cheaper, and more reliable.
However, they have also increased our reliance on computer systems modeling and controlling objects of our physical world.
Prime examples of such objects include cars on our roads, robots in our manufacturing plants, and satellites orbiting our planet. Such systems, referred to as \textit{cyber-physical systems} (CPSs)~\cite{rajkumar2010cyber}, often fall under the category of \textit{hybrid systems}: their programmable controllers typically exhibit discrete behavior, while the laws of physics that the systems are restricted by are continuous in nature. 

The prevalence of hybrid systems around us, coupled with our increased reliance on these systems, has necessitated further exploration of reasoning about such systems. This process involves reasoning about the \emph{states} of the hybrid system. A state is considered \textit{safe} if it does not violate any safety property of the system, and considered \textit{live} if the system can make some useful progress from that state. Verifying a system guarantees safety and liveness in the system. Signal temporal logic (STL)~\cite{maler2004monitoring,maler2013monitoring} was introduced to monitor properties over continuous-time signals of continuous and hybrid systems in given time intervals, and has since been used primarily for monitoring purposes. Dynamic logic~\cite{harel2001dynamic} was introduced as a formal system for reasoning about programs. Differential dynamic logic (d$\mathcal{L}$)~\cite{platzer2008differential} was built on top of dynamic logic to reason about the end states of a hybrid program, to ensure that the end state is a safe state. However, a hybrid system that is in a safe state at the end of a program's execution may not have been in a safe state throughout the program's execution: it is possible for a safety property to be violated during the execution of a program and be held at the termination of the program. 
Therefore, it is vital to verify that hybrid systems are safe during execution in addition to being safe upon termination. Differential temporal dynamic logic (dTL)~\cite{Platzer2010} and differential temporal dynamic logic with nested temporalities (dTL\textsuperscript{2})~\cite{jeannin2014dtl,CMU-CS-14-109} use both dynamic logic -- to reason about all possible executions of a program -- and a fragment linear temporal logic (LTL) -- to reason about intermediate states of each execution -- to tackle this challenge.

While dTL and dTL\textsuperscript{2} are able to reason about intermediate states of a hybrid system during program execution, the logics are still unable to reason about intermediate states of a system in given time intervals. This is a major limitation of the logics, since such reasoning abilities can be crucial in ensuring safety of a hybrid system (e.g., ensuring that a self-driving car applies its brakes within $x$ seconds of spotting a stop sign, as opposed to ensuring that the car applies its brakes \emph{eventually} after spotting a stop sign). STL is able to prove properties about a system in given time intervals, but the logic reasons about only \emph{one} execution of a system, not all possible executions. Therefore, using STL alone to reason about safety in hybrid systems is not sufficient.

In this paper, we present signal temporal dynamic logic (\stdl), a logic that integrates a fragment of STL with differential dynamic logic (d$\mathcal{L}$) to reason both about the intermediate states of a hybrid system \textit{in given time intervals}, and about the final states of the system. This reasoning is enabled by our use of STL, which natively supports formulas of the form $\square_{[a,b]}\phi$ (i.e., for all times between $t+a$ and $t+b$, where $t$ is the current time, the property $\phi$ is true) and $\dia_{[a,b]}\phi$ (i.e., there exists a time between $t+a$ and $t+b$ such that the property $\phi$ is true), but has historically been used mainly for monitoring purposes. We show that STL can be used for full deductive reasoning of hybrid systems.

The main contributions of this work are as follows:
\begin{itemize}
    \item We introduce \stdl -- a logic that reasons about STL formulas for the first time in the context of d$\mathcal{L}$, bringing together results from two different communities with little overlap into a common framework. 
    \item We introduce a notion of timing hybrid programs to bridge the gap between d$\mathcal{L}$ and STL for verification purposes.
    \item We provide a semantics for \stdl and sound proof calculus for the logic, along with a proof of soundness and relative completeness.
\end{itemize}

The rest of the paper is organized as follows. Section~\ref{sec:motivation} motivates \stdl by introducing a running example of a use-case from the industry highlighting the power of the logic. Section~\ref{sec:stdl} introduces the syntax and semantics of \stdl. Section~\ref{sec:proofcalc} motivates the concept of normalization of trace formulas in \stdl and presents the proof system of \stdl. Section~\ref{sec:future} discusses future directions for \stdl. Section~\ref{sec:relatedwork} outlines some related work, and Section~\ref{sec:conclusion} parts with concluding thoughts.
\section{Motivation and Running Example}\label{sec:motivation}
Throughout this paper, we use a simplified example of a use-case for \stdl inspired by industry. As we note in Section~\ref{sec:intro}, a major limitation of the program logics preceding \stdl is their inability to reason about temporal properties in specified time intervals. Such reasoning abilities can be crucial in verifying a hybrid system. While STL is able to handle formulas specifying properties in given time intervals, the logic is only able to prove properties about \emph{one} execution of a hybrid system, and not all possible executions; we need to be able to reason about every execution of a hybrid system to be able to claim with certainty correctness of the system. As such, none of differential dynamic logic, differential temporal dynamic logic, or signal temporal logic -- or other variants of these logics -- alone is sufficient to reason about safety and liveness in hybrid systems.

To see why, let us examine a simplified version of traction assist from the automobile industry. Consider a car with some accelerator input and braking force cruising on the road. The car has a signal that streams a binary value corresponding to whether or not the car's sensors detect that the car is skidding or losing traction, and a Boolean flag corresponding to whether or not the vehicle's traction control is engaged. For simplicity, assume the accelerator input can have a positive or negative value corresponding to acceleration and deceleration respectively, or a value of zero corresponding to no acceleration. Assume further that the braking force is a non-negative integer. Let the wheel rotation of the car's wheels evolve according to some differential equation. The car has a safety property requiring that in the event that the car is skidding, a vehicle traction assist program executes to help gain traction again and slow down the wheel spin to stop the skidding, following which the car can accelerate again. According to the safety property, after running the traction assist program, the car's traction control should turn on and within 1 to 5 seconds, the car's wheel rotation should fall to under some threshold value (to help regain control). 

As we introduce key concepts in the following sections, we also present the differential equation, the hybrid program, and the safety property for the car in \stdl. We present a proof sketch of the safety property using the \stdl calculus. We note that safety properties of this class (i.e., containing temporal references for specified time intervals) are expressible directly in \stdl (but, to the best of our knowledge, not in any other logic preceding \stdl), and remain crucial in verifying correctness of hybrid systems.

\section{Signal Temporal Dynamic Logic}\label{sec:stdl}

This section formally defines the syntax and semantics of hybrid programs and state and trace formulas in \stdl. We take special care to ensure that \stdl is a conservative extension of d$\mathcal{L}$, i.e. the non-temporal aspects of the state semantics for \stdl are equivalent to the non-temporal transition semantics of d$\mathcal{L}$ (Definition 5 in \cite{platzer2008differential}). 

\subsection{Hybrid Programs}\label{subsec:hp}
We use \textit{hybrid programs} to model hybrid systems in our work. A hybrid program $\alpha, \beta$ could be a discrete assignment ($x:= \theta$), a test ($?\chi$), an ordinary differential equation ($x' = \theta \ \& \ \chi$), a non-deterministic choice ($\alpha \cup \beta$), a sequential composition ($\alpha; \beta$), or a non-deterministic finite repetition ($\alpha^*$). As in d$\mathcal{L}$, a term $\theta$ can be any polynomial with a rational coefficient, and a condition $\chi$ can be any first-order formula of real arithmetic.

The syntax of hybrid programs can be summarized as:
$$\alpha, \beta ::= x:= \theta \ | \ ?\chi \ | \ x' = \theta \ \& \ \chi \ | \ \alpha \cup \beta \ | \ \alpha; \beta \ | \ \alpha^*$$

For the semantics of hybrid programs in \stdl, the set of \textit{states} $\sta$ is the set of functions from variables to $\mathbb{R}$. A special state $\Lambda \notin \sta$ denotes a failure state for the hybrid system. The \textit{trace semantics} of a hybrid program $\alpha$ assign a set of traces $\hp{\alpha}$ to the program. For $v \in \sta \cup \{\Lambda\}$, we express the function $\sigma : [0,0] \to \{v\}, 0 \mapsto v$ using $\hat{v}$, and $\hat{v}$ is defined only on the singleton interval [0,0].\footnote{We often informally refer to a trace defined on a singleton interval $\{i\}$, e.g. $(\hat{v})$, as a trace that executes in zero time.} 
A trace, then, is a non-empty, finite sequence $\sigma = (\sigma_0, \sigma_1, \dots, \sigma_n)$ of subtraces $\sigma_i$. For $0\leq i<n$, the piece $\sigma_i$ is a function $\sigma_i : [r_{i-1}, r_i] \to \sta$, with the convention $r_{-1} = 0$, where $r_i - r_{i-1}$ is the duration of this step and $r_i \geq r_{i-1}$. Where $i=n$, $\sigma_n$ can be defined as:
\begin{itemize}
    \item $\sigma_n : [r_{n-1}, r_n] \to \sta$, in which case we refer to $\sigma$ as a \textit{terminating} trace;
    \item $\sigma_n : [r_{n-1}, +\infty) \to \sta$, in which case we refer to $\sigma$ as an \textit{infinite} trace;
    \item $\sigma_n : [r_{n-1}, r_{n-1}] \to \{\Lambda\}$ with $\sigma(r_{n-1}) = \Lambda$, for $n\geq 1$, in which case we refer $\sigma$ as an \textit{error} trace. $n\geq 1$ ensures that $(\hat{\Lambda})$ is not considered as a trace.
\end{itemize}
For a trace $\sigma = (\sigma_0, \dots, \sigma_n)$, we define a \textit{position} of $\sigma$ as a pair $(i, t)$  such that $0\leq i < n$ and $t$ is in the domain of definition of $\sigma_i$. We write $\sigma_i(t)$ to refer to the state of $\sigma$ at $(i, t)$, i.e. $\sigma(i,t) = \sigma_i(t)$, and define the domain of $\sigma$ as:
$${\sf dom}(\sigma) = \bigcup_{i=0}^n \left(\bigcup_{t \in {\sf dom(\sigma_i)}}(i,t)\right)$$
We can now define the lengths of traces of hybrid programs.
\begin{definition}[Length of traces of hybrid programs]\label{def:tracelength}
The length of a trace $\sigma = (\sigma_0, \sigma_1, \dots, \sigma_n) \in \hp{\alpha}$, denoted by $|\sigma|\in\mathbb{R}_+ \cup \{+\infty\}$, is defined inductively as follows:
\begin{itemize}
    \item $|\sigma| = r_n$ if $\sigma_n : [r_{n-1}, r_n] \to \sta$;
    \item $|\sigma| = +\infty$ if $\sigma_n : [r_{n-1}, +\infty) \to \sta$;
    \item $|\sigma| = r_{n-1}$ if $\sigma_n : [r_{n-1}, r_{n-1}] \to \{\Lambda\}$. 
\end{itemize}
\end{definition}

The set of all \textit{traces} of a hybrid program is referred to as $\tra$, and we collectively refer to infinite traces and error traces as non-terminating traces. For a trace $\sigma$, we refer to the state $\sigma_0(0)$ as $\first\sigma$, and we often say that ``$\sigma$ starts with $v$" if $\first\sigma=v$. Likewise, for a finite trace $\sigma$, if $\sigma$ terminates in a non-error state, we refer to the state $\sigma_n(r_n)$ as $\last\sigma$; otherwise, we refer to the state $\Lambda$ as $\last\sigma$. Note that for any trace $\sigma$, $\first\sigma$ is always well-defined, but $\last\sigma$ may not be (since infinite traces have no last state). The value of term $\theta$ in state $v$ is denoted by $val(v,\theta)$, and the valuation assigning variable $x$ to $r \in \mathbb{R}$ while matching with $v$ on all other variables is denoted by $v[x \mapsto r]$. If a state $v$ satisfies some condition $\chi$, we write $v \vDash \chi$; if $v$ does not satisfy condition $\chi$, we write $v \nvDash \chi$. Finally, given a trace $\sigma$ and an $x \in \mathbb{R}$, we use the notation ${\sf dom}(\sigma) \oplus x$ to denote the domain of $\sigma$ shifted by a value of $+x$. For example, if ${\sf dom}(\sigma)=[a,b]$, then ${\sf dom}(\sigma) \oplus x=[a+x,b+x]$.

\begin{definition}[Trace semantics of hybrid programs]\label{def:hybridprograms}
The trace semantics $\hp{\alpha}$ of a hybrid program $\alpha$ is defined as follows: 
\begin{itemize}
    \item $\hp{x:=\theta} = \{(\hat{v}, \hat{w}) \ | \ w = v[x \mapsto val(v,\theta)]\}$;
    \item $\hp{x' = \theta \ \& \ \chi} = \{(\sigma) : \sigma$ is a state flow of order 1~\cite{platzer2008differential} defined on $[0,r]$ or $[0,+\infty)$ solution of $x' = \theta$, and for all $t$ in its definition domain, $\sigma(t) \vDash \chi\} \cup \{(\hat{v}, \hat{\Lambda}) \ | \ v \nvDash \chi\}$;
    \item $\hp{?\chi} = \{(\hat{v}) \ | \ v\vDash\chi\} \cup \{(\hat{v}, \hat{\Lambda}) \ | \ v \nvDash \chi\}$;
    \item $\hp{\alpha \cup \beta} = \hp{\alpha} \cup \hp{\beta}$;
    \item $\hp{\alpha;\beta} = \{\sigma \circ \rho \ | \ \sigma \in \hp{\alpha}, \rho \in \hp{\beta}$ when $\sigma \circ \rho$ is defined$\}$, where the composition $\sigma \circ \rho$ of $\sigma = (\sigma_0, \dots, \sigma_n)$ and $\rho = (\rho_0, \dots, \rho_m)$ is
    \begin{itemize}
    \item $\sigma \circ \rho = (\sigma_0, \dots, \sigma_n, \bar{\rho}_0, \dots, \bar{\rho}_m)$ if $\sigma$ terminates and $\last\sigma=\first\rho$, where $\bar{\rho} = (\bar{\rho}_0, \dots, \bar{\rho}_m)$ is a trace with ${\sf dom}(\bar{\rho}) = {\sf dom}(\rho) \oplus |\sigma|$ and for each $i \in \{0\dots m\}$, for each $t \in {\sf dom}(\rho_i)$, $\bar{\rho}_i(t) = \rho_i(t - |\sigma|)$,\footnote{Informally, $\bar{\rho}$ is merely the trace $\rho$ shifted to the right by a value of $+|\sigma|$.}
    \item $\sigma$ if $\sigma$ does not terminate,
    \item undefined otherwise;
    \end{itemize}
    \item $\hp{\alpha^*} = \bigcup_{n \in \mathbb{N}}\hp{\alpha^n},$ where $\alpha^0$ is defined as ?true, $\alpha^1$ is defined as $\alpha$, and $\alpha^{n+1}$ is defined as $\alpha^n; \alpha$ for $n\geq1$.
\end{itemize}
\end{definition}
These semantics for hybrid programs are adopted from dTL\textsuperscript{2}~\cite{jeannin2014dtl}. As in dTL\textsuperscript{2}, an important property of the trace semantics of hybrid programs is that for any hybrid program $\alpha$ and state $v$, there always exists a trace $\sigma \in \hp{\alpha}$ such that $\first\sigma=v$ (even if $\sigma$ is an error trace). A key difference between the semantics of dTL\textsuperscript{2} and our work is that for a trace $\sigma = (\sigma_0, \dots, \sigma_n)$, while the former define the domain of each $\sigma_i$ from 0 to $r_i$, we define the domain of each $\sigma_i$ from $r_{i-1}$ to $r_i$, to enable easier reasoning about temporal formulas in given time intervals. As such, our semantics for the composition $\sigma \circ \rho$ between traces $\sigma$ and $\rho$ requires trace $\rho$ to be shifted in time by a value of $+|\sigma|$. 

\subsubsection{Running Example: Traction Assist in Cars} Having introduced the semantics of hybrid programs in \stdl, we now formally specify a simplified version of the differential equation ${\sf cruise}$ that varies the car's wheel rotation $\rho$. For $\omega$ the acceleration of the car, $\varphi$ the braking force applied to each of the car's wheels, and some positive constants $k$ and $j$, we have
$${\sf cruise}(\omega, \varphi) ::= \omega \times k - \varphi \times j$$ 
Note that in practice, each of the car's wheels could have a different wheel rotation and braking force. For the sake of simplicity, and to avoid presenting four separate proofs for this example, we assume that each wheel has the same rotation and braking force.

A very simple version of the hybrid program ${\sf traction\_assist}$ can then take the form
\begin{align}
    {\sf traction\_assist} ::=&\ ( ?({\sf no\_traction}); \\
    &\ {\sf traction\_control := 1};\\ 
    &\ \omega := -1;\ \varphi:= 10;\ \rho' = {\sf cruise(\omega, \varphi)});
\end{align}
where the signal ${\sf no\_traction}$ is a binary value of ${\sf true}$ or ${\sf false}$ corresponding to whether or not the car's sensors detect that the car is losing traction and the Boolean flag ${\sf traction\_control}$ keeps track of whether or not the vehicle's traction control is engaged.\footnote{The variable $\varphi$ is set to an arbitrary non-negative integer for the purposes of this example.}

Several properties of hybrid programs are present in the program ${\sf traction\_assist}$. (1) denotes a test to check whether the car has lost traction; (2) represents an assignment statement setting ${\sf traction\_control}$ to on; and (3) shows an evolving ordinary differential equation that changes the wheel rotation of the car. The sequential composition operator joins the individual statements together to form a single hybrid program.

\subsection{State and Trace Formulas}
State and trace formulas are used to reason about hybrid programs. A state formula is used to express properties about a state, whereas a trace formula is used to express properties about a trace.
The syntax of state and trace formulas in \stdl can then be summarized as: 
$$\phi, \psi ::= \theta_1 \geq \theta_2 \ | \  \neg \phi \ | \ \phi \wedge \psi \ | \ \forall x. \phi \ | \ [\alpha] \pi$$
$$\pi ::= \phi \ | \ \neg \pi \ | \ \square_{{[a,b]}} \phi $$ 
$$a,b ::= \theta \ | \ \max(\theta_1, \theta_2) \ | \ \min(\theta_1, \theta_2) \ | \ a + b \ | \ a - b$$

A state formula $\phi$ or $\psi$ could express a comparison of two terms ($\theta_1 \geq \theta_2$), a negation of a state formula ($\lnot \phi$), a conjunction of two state formulas ($\phi \land \psi$), a universally quantified ($\forall x. \phi$) state formula over a variable $x \in \mathbb{R}$, or a program necessity ($\boxp{\alpha}\pi$) indicating that all traces of program $\alpha$ starting from the current state satisfy $\pi$. For a disjunction of two state formulas ($\phi \lor \psi$), we define as an abbreviation $\phi \lor \psi \equiv \lnot(\phi \land \psi)$; for an existentially quantified ($\exists x. \phi$) over a variable $x \in \mathbb{R}$, we define $\exists x. \phi \equiv \lnot \forall x. \lnot \phi$; and for a program possibility ($\diap{\alpha}\pi$) over a trace formula $\pi$ indicating that there exists a trace of program $\alpha$ starting from the current state that satisfies $\pi$, we define $\diap{\alpha}\pi \equiv \lnot \boxp{\alpha}\lnot \pi$.

A trace formula $\pi$ can express a state formula ($\phi$), a negation of a trace formula ($\lnot \pi$), or a temporal necessity ($\square_{[a,b]}\phi$) indicating that given the current time $t$, every trace starting in the current state satisfies $\phi$ from time $t+a$ and $t+b$. A temporal possibility ($\dia_{[a,b]}\phi$) indicating that every trace starting in the current state satisfies $\phi$ at some point between time $t+a$ and time $t+b$ is defined as the abbreviation $\dia_{[a,b]}\phi \equiv \lnot \square_{[a,b]} \lnot \phi$. For time intervals of the form $[a,b]$, $a$ and $b$ are terms in the hybrid program evaluated in the first state of a trace (which is always well-defined, see Definition~\ref{def:tracesatisfaction}), or the $\min$ or $\max$ of two terms in the hybrid program. We allow for $a$ and $b$ to be terms in the hybrid program, and not mere constants, since we need to allow for a program variable to appear as the lower or upper bound of an interval $[a,b]$ (see Section~\ref{sec:timing}, where the timing variable $q$ is introduced to appear inside the temporal intervals of an \stdl formula for interval shifting).

\subsection{Length of Traces and Trace Formulas}
Previous works supporting temporal operators within the context of d$\mathcal{L}$ did not need to reason about the length of a trace or a trace formula, due to their use of linear temporal logic operators that do not support reasoning about formulas in time intervals. However, since \stdl involves verifying a trace over specified time intervals, we need to incorporate reasoning about lengths of traces and trace formulas to determine the satisfaction of formulas over traces of hybrid programs. More specifically, we require that for a hybrid program $\alpha$, a trace $\sigma \in \hp{\alpha}$ needs to be sufficiently long to determine the satisfaction of the program necessities and possibilities. This requirement is inspired by that of STL with respect to signal lengths~\cite{maler2004monitoring,maler2013monitoring}, and is similarly justified for \stdl since it is intuitively nonsensical to verify the satisfaction of a trace formula of length $\varphi$ against a trace of length $\varphi_0 < \varphi$.

\begin{definition}[Minimum length of trace formulas]\label{def:traceformulalength}
 The necessary length associated with trace formula $\pi$, written as $\|\pi\|$, to determine the satisfaction of a program necessity or possibility is defined inductively as follows: 
\begin{align*}
    \| \phi \| &= 0\\
    \| \lnot \pi \| &= \| \pi \|\\
    \|\square_{[a,b]}\phi\| &= b
\end{align*}
\end{definition}


\subsection{Satisfaction of State and Trace Formulas}\label{subsec:formulas_in_stdl}

The satisfaction of state and trace formulas in \stdl is defined as follows:
\begin{definition}[Satisfaction of state formulas]\label{def:statesatisfaction}
For a state formula $\phi$ and state $v \in \sta$, we say $v \vDash \phi$ if $v$ satisfies $\phi$. Satisfaction of state formulas with respect to state $v$ is then defined inductively as follows:
\begin{itemize}
    \item $v \vDash \theta_1 \geq \theta_2$ if and only if $val(v, \theta_1) \geq val(v, \theta_2)$;
    \item $v \vDash \neg \phi$ if and only if $v \nvDash \phi$;
    \item $v \vDash \phi \wedge \psi$ if and only if $v \vDash \phi$ and $v \vDash \psi$;
    \item $v \vDash \forall x. \phi$ if and only if $v[x \mapsto d] \vDash \phi$ for all $d \in \mathbb{R}$;
    \item For $\phi$ a state formula, $v \vDash [\alpha]\phi$ if and only if for every trace $\sigma \in \hp{\alpha}$ such that $\first\sigma = v$, if $\sigma$ terminates, then $\last\sigma \vDash \phi$;
    \item For $\pi$ a trace formula, $v \vDash [\alpha]\pi$ if and only if for every trace $\sigma \in \hp{\alpha}$ such that $\first\sigma = v$, if $|\sigma| \geq val(\first\sigma, \|\pi\|)$, then we also have that $\sigma \vDash \pi$;
    
\end{itemize}
\end{definition}

Definition~\ref{def:statesatisfaction} defines the satisfaction of formulas of the form $\boxp{\alpha}\pi$, for $\pi$ a trace formula, as: ``$v \vDash [\alpha]\pi$ iff for each trace $\sigma \in \hp{\alpha}$ such that $\first\sigma = v$, if $|\sigma| \geq val(\first\sigma, \|\pi\|)$, we also have that $\sigma \vDash \pi$." The choice behind this definition for the semantics is not an obvious one, and as such, is explained here for further clarity.

Since \stdl supports full negation of state formulas, we had take special care to ensure that the property for duals for program modalities is not violated in the logic. One of our utmost concerns was to ensure that for all hybrid programs $\alpha$ and all trace formulas $\pi$, it is always the case that $\boxp{\alpha}\pi \equiv \lnot\diap{\alpha}\lnot \pi$. As such, we had three possible choices for the definition of the semantics for formulas of this form.
\begin{enumerate}
    \item[(i)] $v \vDash [\alpha]\pi$ iff for each trace $\sigma \in \hp{\alpha}$ such that $\first\sigma = v$, we have that $\sigma \vDash \pi$.
    
    To ensure that property for duals holds in this case, we would have to define the dual as:
    
    $v \vDash \diap{\alpha}\pi$ iff there exists a trace $\sigma \in \hp{\alpha}$ such that $\first\sigma = v$ and $\sigma \vDash \pi$.
    
    \item[(ii)] $v \vDash [\alpha]\pi$ iff for each trace $\sigma \in \hp{\alpha}$ such that $\first\sigma = v$, if $|\sigma| \geq val(\first\sigma, \|\pi\|)$, we also have that $\sigma \vDash \pi$.
    
    We would then have to define the dual as:
    
    $v \vDash \diap{\alpha}\pi$ iff there exists a trace $\sigma \in \hp{\alpha}$ such that $\first\sigma = v$, and we have that $|\sigma| \geq val(\first\sigma, \|\pi\|)$ and $\sigma \vDash \pi$
    
    \item[(iii)] $v \vDash [\alpha]\pi$ iff for each trace $\sigma \in \hp{\alpha}$ such that $\first\sigma = v$, we have that $|\sigma| \geq val(\first\sigma, \|\pi\|)$ and $\sigma \vDash \pi$.
    
    We would then have to define the dual as:
    
    $v \vDash \diap{\alpha}\pi$ iff there exists a trace $\sigma \in \hp{\alpha}$ such that $\first\sigma = v$, and if $|\sigma| \geq val(\first\sigma, \|\pi\|)$, we also have that $\sigma \vDash \pi$.
\end{enumerate}

Option (i) is the least complicated and arguably the most intuitive one. However, it has one major limitation: it fails to specify the behavior of the logic when the trace being considered is simply not long enough to determine the satisfaction of a trace formula. Consider the simple hybrid program that $x := 5$. We could have a property that checks this program: $\boxp{x:=5}\square_{[0,10]}(x=5)$. However, recall that a (discrete) trace of assignment terminates in zero time. As such, we are left with the following question: what does it mean for a trace to satisfy a property 10 seconds after it has already terminated? Clearly, we need to consider the length of the trace that the property has to be proven over, and ensure that the trace is of necessary length. This idea is not novel: \cite{maler2004monitoring} uses the same approach for defining satisfaction of formulas over signals.

With option (i) eliminated, we are left with options (ii) and (iii) as the most obvious candidates for the definition of trace semantics of \stdl. Having one of the definitions be an implication and the dual be a conjunction is the only way to ensure that the property for duals holds -- it is not possible to have both definitions be implications or conjunctions. With that in mind, we first look at (iii). It is fairly easy to notice that the definition provided in (iii) make it virtually impossible for $\boxp{\alpha}\pi$ to be true: it requires \emph{every} trace $\sigma \in \hp{\alpha}$ to be of the required length -- a trait that is simply not likely in practice. Similarly, it makes it too easy for $\diap{\alpha}\pi$ to be true: any trace $\sigma \in \hp{\alpha}$ with length $|\sigma|<\|\pi\|$ can trivially satisfy the formula. This leaves option (ii), which provides a definition that makes most sense intuitively. For the $\boxp{\alpha}$ case, it might not be reasonable to require that \emph{all} traces have the required length. But for the $\diap{\alpha}$ case, since the presence of just one satisfying trace is sufficient, it should be the case that that one trace is of the required length. This behavior is captured in the definition in (ii), and we employ that definition in the state and trace semantics of \stdl.

\begin{definition}[Satisfaction of trace formulas]\label{def:tracesatisfaction}
For a trace formula $\pi$ and trace $\sigma = (\sigma_0, \dots, \sigma_n) \in \tra$, we say $(\sigma, (i,t)) \vDash \pi$ if $\sigma$ satisfies $\pi$ starting from subtrace $\sigma_i$ at time $t$. We use $\sigma \vDash \pi$ to say that $(\sigma, (0, 0)) \vDash \pi$. Satisfaction of trace formulas with respect to a trace $\sigma$ is then defined inductively as follows:
\begin{itemize}
    \item For $\phi$ a state formula, $(\sigma, (i,t)) \vDash \phi$ if and only if $|\sigma| \geq t$ and $\sigma_i(t) \vDash \phi$; 
    \item $(\sigma, (i,t)) \vDash \neg \pi$ if and only if $(\sigma, (i,t)) \nvDash \pi$;
    \item $(\sigma,(i,t)) \vDash \square_{[a,b]} \phi$ if and only if for every $t' \in [t+val(\first\sigma,$ $a), t+val(\first\sigma,b)]$ and for every $i$ such that $(i,t') \in {\sf dom}(\sigma)$, it follows that $(\sigma, (i,t')) \vDash \phi$.
\end{itemize}
\end{definition}

Since we define duals as abbreviations, we can build on Definitions~\ref{def:statesatisfaction} and \ref{def:tracesatisfaction} to say that:
\begin{itemize}
    \item $v \vDash \phi \vee \psi$ if and only if $v \vDash \phi$ or $v \vDash \psi$;
    \item $v \vDash \exists x. \phi$ if and only if $v[x \mapsto d] \vDash \phi$ for some $d \in \mathbb{R}$;
    \item For $\phi$ a state formula, $v \vDash \langle\alpha\rangle\phi$ if and only if there exists a trace $\sigma \in \hp{\alpha}$ such that $\sigma$ terminates with $\first\sigma = v$ and $\last\sigma \vDash \phi$;
    \item For $\pi$ a trace formula, $v \vDash \diap{\alpha}\pi$ if and only if there exists a trace $\sigma \in \hp{\alpha}$ such that $|\sigma| \geq val(\first\sigma, \|\pi\|)$ and $\first\sigma = v$ and $\sigma \vDash \pi$;
    \item $(\sigma,(i,t)) \vDash \lozenge_{[a,b]} \phi$ if and only if there exists some $t' \in [t+val(\first\sigma,a), t+val(\first\sigma,b)]$ and there exists some $i$ such that $(i,t') \in {\sf dom}(\sigma)$ and $(\sigma,(i,t')) \vDash \phi$.
\end{itemize}

Given a trace $\sigma$ and an interval $[a,b]$ such that $val(\first\sigma,b) < val(\first\sigma,a)$, we define the interval to be an empty set. As such, formulas such as $\square_{[a,b]}\phi$ and $\dia_{[a,b]}\phi$ are defined to be trivially true and trivially false respectively over this empty interval. This choice deviates from the norm set by STL: formulas like $\square_{[a,b]}\phi$ and $\dia_{[a,b]}\phi$ in STL require that for constants $a$ and $b$, we have $a \geq 0$ and $b \geq a$. This requirement is more difficult to impose in \stdl, since time interval shifting due to sequential composition (see Section~\ref{sec:timing}) could result in a formula where $val(\first\sigma,b) < val(\first\sigma,a)$, and we need the semantics of \stdl to handle such cases appropriately. In the rest of the paper, given a trace $\sigma$ and an interval $[a,b]$, we refer to $val(\first\sigma,a)$ and $val(\first\sigma,b)$ as simply $a$ and $b$ respectively for easier readability.

\subsection{Timing Hybrid Programs}\label{sec:timing}
A major technical difficulty arising from our integration of STL with d$\mathcal{L}$ is the fact that we now need to reason about not only the time intervals where a certain temporal property holds, but also about how the length of a trace of a hybrid program affects the time intervals under consideration. This problem surfaces immediately for the sequential composition of two programs $\alpha$ and $\beta$, but is in fact a general challenge with the integration of continuous traces from d$\mathcal{L}$ and temporal operators from STL.

Let us consider a trace $\sigma = \sigma_\alpha \circ \sigma_\beta \in \hp{\alpha;\beta}$ such that $\sigma_\alpha \in \hp{\alpha}$ terminates at time $c$, following which $\sigma_\beta \in \hp{\beta}$ begins. For simplicity, let us also assume that $a<c<b$ in determining the satisfiability of $\square_{[a,b]}\psi$ by $\sigma$. Note that $\sigma \vDash \square_{[a,b]}\psi$ if and only if $\sigma_\alpha \vDash \square_{[a, c]}\psi$ and $\sigma_\beta \vDash \square_{[0, b-c]}\psi$. Intuitively, this means that $\alpha$ runs first \emph{until time $c$} and $\sigma_\alpha$ satisfies $\psi$ from time $t+a$ to time $t+c$ (where $t$ is the current time), following which $\beta$ runs and $\sigma_\beta$ satisfies $\psi$ from the time it starts to the time $t+(b-c)$ (due to a shifting of the time interval, since part of the interval $[a,b]$ was already satisfied by $\sigma_\alpha$). A key property that this rule relies on is the termination of program $\alpha$ at time $c$. The value of $c$ is not known by a programmer in advance (since a program can have non-deterministic properties), although a programmer could annotate the code to enforce the termination of a program at a certain time. A more elegant solution, however, is to measure the time it takes for a program $\alpha$ to run, and use the measured value for the time offset for any subsequent temporal operators that may need interval shifting.
\begin{definition}[Timing of hybrid programs]\label{def:timing}
Given hybrid programs $\alpha$ and $\beta$, and a  variable $q$ \emph{fresh} in $\alpha$ and $\beta$, the timing of hybrid programs is defined inductively as follows:
\begin{itemize}
    \item $time(x:=\theta) \triangleq x:=\theta$
    \item $time(x' = \theta \ \& \ \chi) \triangleq \{x'=\theta, q'=1 \ \& \ \chi\}$ 
    \item $time(?\chi) \triangleq \ ?\chi$
    \item $time(\alpha \cup \beta) \triangleq time(\alpha) \cup time(\beta)$
    \item $time(\alpha; \beta) \triangleq time(\alpha); time(\beta)$
    \item $time(\alpha^*) \triangleq (time(\alpha))^*$
\end{itemize}
The time taken $q$ by a hybrid program $\alpha$ is then given by the program:
$$timed(\alpha,q) \equiv q:=0;\ time(\alpha)$$
\end{definition}

Recall that a trace $\sigma$ is a function that maps a pair $(i,t)$ to a state $v \in \sta$, whereas a state $v$ is a function from the set of variables $\var$ to $\mathbb{R}$. For $\sigma_\alpha \in \hp{\alpha}$, we write $\sigma_\alpha|_{S}$ to refer to $\sigma$ restricted to variables in the set $S \subseteq \var$. Mathematically, $\sigma_\alpha|_S : (\mathbb{R} \times \mathbb{N}) \rightarrow S \rightarrow \mathbb{R}$, where $(i,t) \mapsto \sigma_\alpha(i,t)|_S$. We can then define an equality between timed and untimed hybrid programs as follows:

\begin{lemma}[Equality of timed and untimed hybrid programs]\label{lemma:timedprograms}
Given a hybrid program $\alpha$, the following set equality always holds:
\begin{align*}
    \{\sigma_\alpha|_{\var - \{q\}} : \sigma_\alpha \in \hp{\alpha}\} &=\\ \{\sigma_{timed(\alpha,q)}|_{\var - \{q\}} &: \sigma_{timed(\alpha,q)} \in \hp{timed(\alpha,q)}\}
\end{align*}
\end{lemma}
\begin{proof}
The proof of Lemma~\ref{lemma:timedprograms} is true by Definition~\ref{def:timing}, keeping in mind the fact that $q$ is fresh in $timed(\alpha,q)$. 
\end{proof}
Intuitively, Lemma~\ref{lemma:timedprograms} expresses that for a trace $\sigma_\alpha \in \hp{timed(\alpha,q)}$, there always exists a corresponding trace $\sigma_\alpha' \in \hp{\alpha}$, and vice versa, such that $\sigma_\alpha$ and $\sigma_\alpha'$ are identical with respect to every variable except the fresh variable $q$ introduced by $timed(\alpha,q)$. We rely on this lemma for the proof of soundness of the \stdl calculus.

\section{Proof Calculus}\label{sec:proofcalc}
In this section, we outline a proof calculus for \stdl, and present a proof of soundness for the rules in the schemata of the proof calculus. 
\subsection{Normalization of Trace Formulas}\label{sec:normalization}
Sequential composition of two traces is a major challenge in a calculus handling alternating program and temporal modalities. To see why, let us consider a state formula $\diap{\alpha;\beta}\square_{[a,b]}\phi$ limited to terminating traces only for simplicity. This formula states that there exists a trace $\sigma_\alpha\in\hp{\alpha}$ followed by the trace $\sigma_\beta\in\hp{\beta}$ such that sequential composition of the traces satisfies $\square_{[a,b]}\phi$. Let us assume further for simplicity that all traces $\sigma_\alpha\in\hp{\alpha}$ terminate between time $a$ and $b$. A first attempt at writing a rule for this state formula could take the form:
$$\frac{\diap{\alpha}\square_{[a,b]}\phi \land \diap{timed(\alpha,q)}\diap{\beta}\square_{[0,b-q]}\phi}{\diap{\alpha;\beta}\square_{[a,b]}\phi}$$
Unfortunately, this rule is intuitive but not sound, since the choice of $\sigma_\alpha$ and $\sigma_\beta$ could be non-deterministic. The premise says that there exists a trace $\sigma_\alpha \in \hp{\alpha}$ in which $\square_{[a,b]}\phi$ is true, and a trace trace $\sigma_\alpha' \in \hp{\alpha}$ followed by $\sigma_\beta$ in which $\square_{[0,b-q]}\phi$ is true, but $\sigma_\alpha$ and $\sigma_\alpha'$ need not necessarily be the same trace. To capture the fact that $\sigma_\alpha$ and $\sigma_\alpha'$ are indeed the same traces, we need a premise resembling:
$$\diap{timed(\alpha,q)}(\square_{[a,b]}\phi \land \diap{\beta}\square_{[0,b-q]}\phi)$$
The rule is not in the syntax of \stdl, since it involves a conjunction between a state formula and a trace formula. We could choose to add this conjunction to the syntax of the logic, but we would still need to reason about the meaning of this conjunction if the trace $\sigma_\alpha$ is non-terminating.

To circumvent this problem cleanly, we need a conjunction operator that reasons about properties like $\phi$ that are true at the end of a trace and properties like $\square_{[a,b]}\phi$ that are true during a trace. dTL\textsuperscript{2} introduces a notion of normalized trace formulas to achieve the expressibility needed for sequential composition for LTL formulas within the context of hybrid systems by introducing a conjunction operator $\sqand$ and a disjunction operator $\sqor$~\cite{jeannin2014dtl}. We extend \stdl with a similar normalization of trace formulas to reason about time-bounded trace properties during the execution of a trace and state properties at the end of a trace. 
We augment the syntax of state formulas to accept normalized trace formulas, and define the syntax of a normalized trace formula $\xi$ as:
\begin{align*}
    \phi, \psi ::=&\ \dots \ | \ [\alpha] \xi \ | \ \langle\alpha\rangle \xi\\
    \xi ::=& \ \phi \sqcap \square_{[a,b]} \psi \ | \ \phi \sqcup \lozenge_{[a,b]} \psi  
\end{align*}

\begin{definition}[Semantics of normalized trace formulas]\label{def:normalization}
For a normalized trace formula $\xi$ and trace $\sigma = (\sigma_0, \dots, \sigma_n) \in \tra$, we say $(\sigma, (i,t)) \vDash \xi$ if $\sigma$ satisfies $\xi$ starting from subtrace $\sigma_i$ at time $t$. We say that $\sigma \vDash \xi$ if $(\sigma, (0,0)) \vDash \xi$. Satisfaction of normalized trace formulas with respect to a trace $\sigma$ is then defined inductively as follows:
\begin{itemize}
    \item $\sigma \vDash \phi \sqcap \square_{[a,b]} \psi$ if and only if
    \begin{itemize}
        \item $\last\sigma \vDash \phi$ and $\sigma \vDash \square_{[a,b]} \psi$, if $\sigma$ terminates,
        \item $\sigma \vDash \square_{[a,b]} \psi$ otherwise;
    \end{itemize}
    
    \item $\sigma \vDash \phi \sqcup \lozenge_{[a,b]} \psi$ if and only if
    \begin{itemize}
        \item $\last\sigma \vDash \phi$ or $\sigma \vDash \lozenge_{[a,b]} \psi$, if $\sigma$ terminates,
        \item $\sigma \vDash \lozenge_{[a,b]} \psi$ otherwise.
    \end{itemize}
\end{itemize}
\end{definition}

Given a normalized state formula $\xi$, we use the notation $\xi_{sta}$ to refer to the state formula in $\xi$, and we use the notation $\xi_{tra}$ to refer to the trace formula in $\xi$. For example, $(\phi \sqand \square_{[a,b]}\psi)_{sta} = \phi$, and $(\phi \sqand \square_{[a,b]}\psi)_{tra} = \square_{[a,b]}\psi$. 

We define the minimum length of normalized trace formulas required to determine the satisfaction of program necessities and possibilities as follows:
\begin{definition}[Minimum length of normalized trace formulas]
The minimum length associated with a normalized trace formula $\xi$, denoted by $\|\xi\|$, to determine the satisfaction of a program necessity or possibility is defined as follows:
\begin{align*}
    \|\phi \sqand \square_{[a,b]}\psi\| &= b\\
    \|\phi \sqor \dia_{[a,b]}\psi\| &= b
\end{align*}
\end{definition}

We build on Definition~\ref{def:statesatisfaction} for state formulas as follows:
\begin{itemize}
    \item $v\vDash\boxp{\alpha}\xi$ if and only if for each trace $\sigma \in \hp{\alpha}$ such that $\first\sigma=v$ and if $\sigma$ terminates then $\last\sigma \vDash \xi_{sta}$, and if $|\sigma| \geq val(\first\sigma, \|\xi\|)$, we also have that $\sigma \vDash \xi$;
    \item $v\vDash\diap{\alpha}\xi$ if and only if there exists trace $\sigma \in \hp{\alpha}$ such that $\first\sigma=v$ and if $\sigma$ terminates then $\last\sigma \vDash \xi_{sta}$, and $|\sigma| \geq val(\first\sigma, \|\xi\|)$ and $\sigma \vDash \xi$.
\end{itemize}

Given the semantics of normalized trace formulas in \stdl, we derive rules to transform any trace formula in \stdl into a normalized trace formula. The rules for normalization are shown in Figure~\ref{fig:normalization}. The relation $\leadsto$ allows us to only consider normalized trace formulas for the rules of the proof calculus of \stdl, thereby simplifying the proof system greatly.
\begin{figure}[t]
    \centering
    \begin{align*}
        &\phi \leadsto \phi \ &(\leadsto\phi)\\
        &\square_{[a,b]}\phi \leadsto \true\sqand\square_{[a,b]}\phi \ &(\leadsto\square_I)\\ 
        &\dia_{[a,b]}\phi \leadsto \false\sqor\dia_{[a,b]}\phi \ &(\leadsto\dia_I)
    \end{align*}
    \caption{Normalization of trace formulas in \stdl.}
    \label{fig:normalization}
\end{figure}

\begin{lemma}[Soundness of normalized trace formulas]\label{lemma:soundnessnormalization}
If $\pi \leadsto \xi$, then for all traces $\sigma$, it follows that $\sigma \vDash \pi$ if and only if $\sigma \vDash \xi$.
\end{lemma}
\begin{proof}
Soundness of rule $(\leadsto\phi)$ is trivial. Soundness of rules $(\leadsto\square_I)$, $(\leadsto\dia_I)$ is true by the semantics in Definition~\ref{def:normalization}.
\end{proof}

\begin{lemma}[Existence of a normalized trace formula]\label{lemma:normalizedexistence}
For any trace formula $\pi$, there exists a state formula $\phi$ such that $\pi \leadsto \phi$, or a normalized trace formula $\xi$ such that $\pi \leadsto \xi$.
\end{lemma}
\begin{proof}
This lemma is a consequent of the $\leadsto$ relation presented in Figure~\ref{fig:normalization}.
\end{proof}

Lemma~\ref{lemma:normalizedexistence} allows the proof system of \stdl to just focus on axiomatizing only formulas that use normalized traces, and inherit non-temporal rules from d$\mathcal{L}$~\cite{platzer2010logical,platzer2012logics,platzer2008differential}. This results in a cleaner, simpler proof calculus for \stdl.

\subsubsection{Running Example: Traction Assist in Cars} Recall that our running example introduced a safety property, $\phi$, that required a skidding car's traction assist to reduce the wheel rotation $\rho$ of the car to some constant, $\rho_0$, within 1 to 5 seconds to help regain traction. This property can be expressed as a normalized \stdl formula as follows:
$$\phi ::= {\sf \boxp{traction\_assist}} 
\left(
\begin{tabular}{c}
    $\lnot{\sf no\_traction}$ $\sqor$ $\dia_{[1, 5]}(\rho < \rho_0)$
\end{tabular}
\right)$$
For ease of understanding, the normalized disjunction $\lnot{\sf no\_traction}$ $\sqor$ $\dia_{[1, 5]}(\rho < \rho_0)$ can be thought of as the implication ${\sf no\_traction}$ $\Rightarrow$ $\dia_{[1, 5]}(\rho < \rho_0)$ (although this implication is not directly supported in the sytax of \stdl). We provide a proof sketch of this property in Section~\ref{sec:casestudyproof}.

\subsection{Proof Calculus}\label{sec:proofrules}
This section presents the proof calculus of \stdl. As in d$\mathcal{L}$, the rules in the proof calculus of \stdl typically follow a symbolic decomposition pattern whereby hybrid programs may be decomposed syntactically as needed. The proof calculus transforms STL formulas into temporal-free formulas to leverage the non-temporal rules of d$\mathcal{L}$. As such, the proof system inherits its non-temporal rules from d$\mathcal{L}$~\cite{platzer2008differential,platzer2012logics,platzer2010logical}, and adds its own temporal rules to allow for expressing temporal formulas for given time intervals. All rules should be used in the same way as in the d$\mathcal{L}$ proof calculus.

Note that with the exceptions of rules (ind $\sqor_t$) and (con $\sqand_t$) (see Figure 2), all rules are actually equivalences between the premise and the conclusion. In other words, each rule has a dual such that the negation of both the premise and the conclusion is also true. Therefore, when we write rule $\displaystyle \frac{\rho}{\phi}$, the following two rules are both true:
$$\frac{\rho}{\phi}\qquad\qquad \frac{\lnot\rho}{\lnot\phi}$$ 
Such duals for the rules contain the proof rules for $\langle\alpha\rangle$ when the original rule contains the proof rules for $\boxp{\alpha}$, and vice versa (again, except for rules (ind $\sqor_t$) and (con $\sqand_t$)).

\subsubsection{Inheritance of Non-Temporal and Temporal Rules}
In addition to the temporal rules introduced in Figure 2, \stdl also uses the proof system of d$\mathcal{L}$. Indeed, the goal of the proof calculus introduced here is to leverage the non-temporal rules of d$\mathcal{L}$ to reason about temporal properties of formulas. Since we build \stdl to conservatively extend d$\mathcal{L}$, it is sound to inherit the proof calculus of d$\mathcal{L}$. 

\subsubsection{Introduction of New Temporal Rules}\label{sec:inferencerules}
This subsection introduces the temporal rules, grouped by program construct for hybrid programs, for the proof calculus of \stdl. A detailed rule schemata for the proof calculus is included in Figure 2.

\begin{figure} \label{fig:proof_calculus}
\begin{center}
\raggedright
\scalebox{1.1}{\textbf{Normalization of Trace Formulas}}
$$\frac{\pi \leadsto \xi \ \ \boxp{\alpha}\xi}{\boxp{\alpha}\pi}\ (\boxp{\ }\leadsto) \ \ \ 
\frac{\pi \leadsto \xi \ \ \diap{\alpha}\xi}{\diap{\alpha}\pi} \ (\diap{}\leadsto)$$
\\
\scalebox{1.1}{\textbf{Assignment}}
$$\frac{\left(
\begin{tabular}{c}
$(a=0 \land b=0)\land(\psi\land\boxp{x:=\theta}(\phi\land\psi))\ \lor$\\
$((a>0 \land b\geq a)\land\boxp{x:=\theta}\phi)$
\end{tabular}
\right)}
{[x := \theta](\phi \sqcap \square_{[a,b]} \psi)}\ ([:=]\ \sqcap_{t})$$
$$\frac{\left(
\begin{tabular}{c}
$(a=0 \land b=0)\land(\psi\lor\boxp{x:=\theta}(\phi\lor\psi))\ \lor$\\
$((a>0 \land b\geq a)\land\boxp{x:=\theta}\phi)$
\end{tabular}
\right)}
{[x := \theta](\phi \sqcup \lozenge_{[a,b]} \psi)}\ ([:=]\ \sqcup_{t})$$
$$\frac{
\left(
\begin{tabular}{c}
$(((a=0 \land b=0)\land(\psi\lor\diap{x:=\theta}(\phi\lor\psi))\ \lor$\\
$((a>0 \land b\geq a)\land\diap{x:=\theta}\phi$
\end{tabular}
\right)
}
{\diap{x := \theta}(\phi \sqcup \lozenge_{[a,b]} \psi)}\ (\diap{:=}\ \sqcup_{t})$$
$$\frac{
\left(
\begin{tabular}{c}
$((a=0 \land b=0)\land(\psi\land\diap{x:=\theta}(\phi\land\psi))\ \lor$\\
$((a>0 \land b\geq a)\land\diap{x:=\theta}\phi)$
\end{tabular}
\right)
}
{\diap{x := \theta}(\phi \sqcap \square_{[a,b]} \psi)}\ (\diap{:=}\ \sqcap_{t})$$
\\
\scalebox{1.1}{\textbf{Test}}
$$\frac{
\left(
\begin{tabular}{c}
$((a=0\land b=0)\land((\chi \wedge (\phi \wedge \psi)) \vee (\neg \chi \wedge \psi))\ \lor$\\
$((a>0\land b\geq a)\land(\lnot\chi \lor (\chi \land \phi)))$
\end{tabular}
\right)}
{[?\chi](\phi \sqcap \square_{[a,b]} \psi)}\ ([?] \ \sqcap_t)$$
$$\frac{
\left(
\begin{tabular}{c}
$((a=0\land b=0)\land((\chi \wedge (\phi \lor \psi)) \vee (\neg \chi \wedge \psi))\ \lor$\\
$((a>0\land b\geq a)\land(\lnot\chi \lor (\chi \land \phi)))$
\end{tabular}
\right)}
{[?\chi](\phi \sqcup \dia_{[a,b]} \psi)}\ ([?] \ \sqcup_t)$$
$$\frac{
\left(
\begin{tabular}{c}
$((a=0\land b=0)\land((\chi \wedge (\phi \lor \psi)) \vee (\neg \chi \wedge \psi))\ \lor$\\
$((a>0\land b\geq a)\land(\lnot\chi \lor (\chi \land \phi)))$
\end{tabular}
\right)}
{\diap{?\chi}(\phi \sqcup \dia_{[a,b]} \psi)}\ (\diap{?} \ \sqcup_t)$$
$$\frac{
\left(
\begin{tabular}{c}
$((a=0\land b=0)\land((\chi \wedge (\phi \wedge \psi)) \vee (\neg \chi \wedge \psi))\ \lor$\\
$((a>0\land b\geq a)\land(\lnot\chi \lor (\chi \land \phi)))$
\end{tabular}
\right)}
{\diap{?\chi}(\phi \sqcap \square_{[a,b]} \psi)}\ (\diap{?} \ \sqcap_t)$$
\\

\scalebox{1.1}{\textbf{Non-deterministic Choice}}
$$\frac{\boxp{\alpha}\xi \land \boxp{\beta}\xi}{\boxp{\alpha \cup \beta}\xi}\ ([\cup]\ \xi) \ \ \ \ \ \ \ \ \ \ \ 
\frac{\diap{\alpha}\xi \lor \diap{\beta}\xi}{\diap{\alpha \cup \beta}\xi}\ (\diap{\cup}\ \xi)$$
\\
\scalebox{1.1}{\textbf{Sequential Composition}}
$$\frac{
[timed(\alpha,q)]([\beta](\phi \sqcap \square_{[\max(0,a-q),b-q]} \psi)\sqcap \square_{[a,\min(b,q)]} \psi)
}
{\boxp{\alpha;\beta}(\phi\sqand\square_{[a,b]}\psi)}\ ([;]\sqand_t)$$
$$\frac{[timed(\alpha,q)]([\beta](\phi \sqcup \dia_{[\max(0,a-q),b-q]} \psi)\sqcup \dia_{[a,\min(b,q)]} \psi)}{\boxp{\alpha;\beta}(\phi\sqor\dia_{[a,b]}\psi)}\ ([;]\sqor_t)$$
$$\frac{\diap{timed(\alpha,q)}(\diap{\beta}(\phi \sqcup \dia_{[\max(0,a-q),b-q]} \psi)\sqcup \dia_{[a,\min(b,q)]} \psi)}{\diap{\alpha;\beta}(\phi\sqor\dia_{[a,b]}\psi)}\ (\diap{;}\sqor_t)$$
$$\frac{\diap{timed(\alpha,q)}(\diap{\beta}(\phi \sqcap \square_{[\max(0,a-q),b-q]} \psi)\sqcap \square_{[a,\min(b,q)]} \psi)}{\diap{\alpha;\beta}(\phi\sqand\square_{[a,b]}\psi)}\ (\diap{;}\sqand_t)$$

\end{center}
\caption{Rule schemata of the proof calculus for \stdl.}
\end{figure}

\begin{figure} \label{fig:proof_calculus}
\ContinuedFloat
\begin{center}
\raggedright
\scalebox{1.1}{\textbf{Ordinary Differential Equation}}
$$\frac{
\left(
\begin{tabular}{c}
    $(b < a \land [x'=\theta \ \& \ \chi]\phi) \ \lor$\\
    $(\lnot \chi \land (\lnot(a=0 \land b \geq a) \lor \psi)) \ \lor$ \\
    $([t:=0;\{x'=\theta, t'=1 \ \& \ (\chi \land t \leq a)\}; ?(t=a)]$\\
    $[\{x'=\theta, t'=1 \ \& \ (\chi \land t \leq b)\}]\psi \land [x'=\theta \ \& \ \chi]\phi)$
\end{tabular}
\right)
}{[x'=\theta \ \& \ \chi](\phi \sqand \square_{[a,b]} \psi)}\text{ ($['] \ \sqcap_{t}$)}$$
$$\frac{
\left(
\scalebox{0.94}{\begin{tabular}{c}
    $(b < a \land [x'=\theta \ \& \ \chi]\phi) \ \lor$\\
    $((\chi \lor [t := 0; \{x' = \theta, t' = 1\ \&\ (\chi \land t \leq a)\}; ?(t=a)]\psi)\ \land$ \\
    $\boxp{x' = \theta \ \& \ (\chi \land \lnot\psi)}\phi\ \land$\\
    $[t := 0; \{x' = \theta, t' = 1\ \&\ (\chi \land t \leq a)\}; ?(t=a)]$\\
    $\langle\{x' = \theta, t'=1\ \&\ (t \leq b)\}\rangle(\lnot \chi \lor \psi))$
\end{tabular}}
\right)
}{[x'=\theta \ \& \ \chi](\phi \sqor \dia_{[a,b]} \psi)}\text{ ($['] \ \sqcup_{t}$)}$$
$$\frac{
\left(
\begin{tabular}{c}
    $(b < a \land [x'=\theta \ \& \ \chi]\phi) \ \lor$\\
    $(\lnot \chi \land (\lnot(a=0 \land b \geq a) \lor \psi)) \ \lor$ \\
    $([t:=0;\{x'=\theta, t'=1 \ \& \ (\chi \land t \leq a)\}; ?(t=a)]$\\
    $[\{x'=\theta, t'=1 \ \& \ (\chi \land t \leq b)\}]\psi \land [x'=\theta \ \& \ \chi]\phi)$
\end{tabular}
\right)
}{[x'=\theta \ \& \ \chi](\phi \sqand \square_{[a,b]} \psi)}\text{ ($['] \ \sqcap_{t}$)}$$
$$\frac{
\left(
\scalebox{0.93}{
\begin{tabular}{c}
    $(b < a \land [x'=\theta \ \& \ \chi]\phi) \ \lor$\\
    $((\chi \lor [t := 0; \{x' = \theta, t' = 1\ \&\ (\chi \land t \leq a)\}; ?(t=a)]\psi)\ \land$ \\
    $\boxp{x' = \theta \ \& \ (\chi \land \lnot\psi)}\phi\ \land$\\
    $[t := 0; \{x' = \theta, t' = 1\ \&\ (\chi \land t \leq a)\}; ?(t=a)]$\\
    $\langle\{x' = \theta, t'=1\ \&\ (t \leq b)\}\rangle(\lnot \chi \lor \psi))$
\end{tabular}}
\right)
}{[x'=\theta \ \& \ \chi](\phi \sqor \dia_{[a,b]} \psi)}\text{ ($['] \ \sqcup_{t}$)}$$
\\
\scalebox{1.1}{\textbf{Non-deterministic Finite Repetition}}
$$\frac{\left(
\begin{tabular}{c}
    $(\phi\land(\lnot(a=0 \land b=0) \lor \psi)) \ \land$\\
    $[timed(\alpha^*,q)][\alpha](\phi \sqcap \square_{[max(0,a-q),b-q]} \psi)$ 
\end{tabular}
\right)}
{[\alpha^*](\phi \sqcap \square_{[a,b]} \psi)}\text{ ([*] $\sqcap_t$)}$$
$$\frac{(\psi \land (a=0\land b\geq a)) \lor (\phi \wedge [\alpha^*; \alpha](\phi \sqcup \lozenge_{[a,b]} \psi))}{[\alpha^*](\phi \sqcup \lozenge_{[a,b]} \psi)}\text{ ([$^{*n}$] $\sqcup_t$)}$$
$$\frac{\phi \implies [\alpha](\phi \sqcup \lozenge_{[a,b]} \psi)}{\forall^\alpha(\phi \implies [\alpha^*](\phi \sqcup \lozenge_{[a,b]} \psi))}\text{ (ind $\sqcup_t$)}$$
$$\frac{\forall^\alpha \forall r > 0 \ (\varphi(r) \implies \diap{\alpha}(\varphi(r-1)\sqand \square_{[a,b]} \psi))}{(\exists r.\varphi(r)) \wedge \psi \implies \diap{\alpha^*}((\exists r. r \leq 0 \land \varphi(r)) \sqand \square_{[a,b]} \psi)} \ (\textrm{con } \sqand_t)$$
$$\frac{
\left(
\begin{tabular}{c}
$(\psi \land (a=0\land b\geq a))\ \lor$\\
$\diap{timed(\alpha^*,q)}\diap{\alpha}(\phi \sqcup \dia_{[max(0,a-q),b-q]} \psi)$
\end{tabular}
\right)
}
{\diap{\alpha^*}(\phi \sqcup \dia_{[a,b]} \psi)}\text{ $(\diap{^*} \sqcup_t)$}$$
$$\frac{
\left(
\begin{tabular}{c}
$(\phi\land(\lnot(a=0 \land b=0) \lor \psi))\ \land$\\
$(\phi \lor \diap{\alpha^*; \alpha}(\phi \sqcap \square_{[a,b]} \psi))$
\end{tabular}
\right)
}
{\diap{\alpha^*}(\phi \sqcap \square_{[a,b]} \psi)}\text{ ($\diap{^{*n}}$ $\sqcap_t$)}$$

\end{center}
\caption*{Figure 2  (continued): Rule schemata of the proof calculus for \stdl.}
\end{figure}

Rules $(\boxp{\ }\leadsto)$ and $(\diap{}\leadsto)$ lift normalization of trace formulas to program necessities and possibilities respectively.

For assignment rule ([$:=]\ \sqcap_{t}$), the first disjunct expresses that for the time interval $[0,0]$, $\psi$ must hold initially, and after the execution of the program, must continue to hold in addition to $\phi$, as summarized in clause $\psi \wedge [x := \theta](\phi \wedge \psi)$. The second disjunct expresses that for any interval $[a,b]$ where $a > 0$ and $b \geq a$, only $\phi$ needs to be true after execution of the assignment, since assignment occurs in zero time, and as such, the trace of $x:=\theta$ would not be long enough to determine the satisfiability of $\square_{[a,b]}\psi$ for $a>0$. Similar reasoning is used for rule $(\textrm{[:=]} \ \sqcup_{t})$.

For the rules for test, as a reminder, a test trace only terminates if the test passes, and is a trace of the error state if the test fails. Rule ($[?] \ \sqcap_t$) encapsulates the fact that a trace of $?\chi$ satisfies $\phi \sqcap \square_{[a,b]} \psi$ if and only if 
\begin{itemize}
    \item for $a=0$ and $b=0$, its initial state satisfies $\phi \land \psi$ if the test passes, or satisfies only $\psi$ if the test fails;
    \item for $a>0$ and $b \geq a$, its initial state satisfies just $\phi$ if the test passes.
\end{itemize}
Note that there is no satisfaction requirement on the trace of a failing test (i.e., $\lnot \chi$ is true) when $a>0$ and $b\geq a$, since the test also occurs in zero time, and as such, the trace of $?\chi$ would not be long enough to determine the satisfiability of $\square_{[a,b]}\psi$ in this case. Similar reasoning is used for rule ($[?] \ \sqcup_t$).

Rules for ordinary differential equations (ODEs) look complex at first glance, but can be broken down in slightly simpler sub-rules. It is first important to remember that ODEs could have terminating traces or error traces, and the rules for ODEs need to account of both possibilities. With that in mind, we conclude that an error trace of $x'=\theta \ \& \ \chi$ satisfies $\phi \sqand \square_{[a,b]} \psi$ if and only if $a=0$ and $b\geq a$ implies $\psi$, as the second disjunct in rule ($['] \ \sqcap_{t}$). For non-error traces of $x'=\theta \ \& \ \chi$, we first transform the program into a program of the form $t:=0;\{x'=\theta, t'=1 \ \& \ (\chi \land t \leq a)\}; ?(t=a)$ and $\{x'=\theta, t'=1 \ \& \ (\chi \land t \leq b)\}$ to enforce that the differential equation first runs from time $t=0$ to time $t=a$ without any satisfaction requirements on $\psi$, followed by running the equation from time $t=a$ to $t=b$, during which $\psi$ must be true. In addition to this, $\phi$ must be true after running the program $x'=\theta \ \& \ \chi$, to deal with the case where the execution exits the differential equation before time $a$. This is summarized in the third disjunct of the rule ($['] \ \sqcap_{t}$). Note that the first disjunct of the rule deals with the case where $b<a$, so $\square_{[a,b]}\psi$ is defined to be trivially true, and any trace of $x'=\theta \ \& \ \chi$ need only satisfy $\phi$. Rule ($['] \ \sqcup_{t}$) expresses that a trace of $x'=\theta \ \& \ \chi$ satisfies $\phi \sqor \dia_{[a,b]}\psi$ if and only if either $b<a$ and the trace satisfies $\phi$ upon termination, or
\begin{itemize}
    \item the differential equation can evolve or has satisfied $\psi$ at time $t=a$ (as in the first conjunct of the rule), 
    \item if no trace of the differential equation can satisfy $\psi$, all traces must satisfy $\phi$ instead (as in the second conjunct of the rule), 
    \item either there does not exist a non-terminating trace of the differential equation -- transformed to a program as in rule ($['] \ \sqcap_{t}$) -- or such a trace satisfies $\psi$ between times $a$ and $b$.
\end{itemize}

Rule $([\cup]\ \xi)$ for non-deterministic choice is lifted directly from the corresponding rule $[\cup]\square$ in d$\mathcal{L}$.

\begin{figure*}[t]
    \centering
    \begin{prooftree}
        \AxiomC{\dots}
        \UnaryInfC{$\boxp{timed(\alpha_1,q_1)}(\boxp{timed(\alpha_2,q_2)}(\psi_3\sqor\dia_{[\max(0,1-q_1),\min(5-q_1,q_2)]} )\sqor \dia_{[1,\min(5,q_1)]}\psi_2)$}
        \RightLabel{$\boxp{'}\sqor_t$}
        \UnaryInfC{$\boxp{timed(\alpha_1,q_1)}(\boxp{timed(\alpha_2,q_2)}(\boxp{\alpha_3}(\psi_1 \sqor \dia_{[\max(0,\max(0,1-q_1)-q_2),5-q_1-q_2]}\psi_2)\sqor\dia_{[\max(0,1-q_1),\min(5-q_1,q_2)]} )\sqor \dia_{[1,\min(5,q_1)]}\psi_2)$}
        \RightLabel{$\boxp{;}\sqor_t$}
        \UnaryInfC{$\boxp{timed(\alpha_1,q_1)}(\boxp{\alpha_2;\alpha_3}(\psi_1 \sqor \dia_{[\max(0,1-q_1),5-q_1]}\psi_2)\sqor \dia_{[1,\min(5,q_1)]}\psi_2)$}
        \RightLabel{$\boxp{;}\sqor_t$}
        \UnaryInfC{$\boxp{\alpha}(\psi_1 \sqor \dia_{[1,5]}\psi_2)$}
    \end{prooftree}
    \caption{Proof sketch of our running example. The state formula $\psi_3$ can be further proven using rules from d$\mathcal{L}$.}
    \label{fig:case_study_proofsketch}
\end{figure*}

The rules for sequential composition were one of the most challenging aspects of \stdl. Indeed, sequential composition is the sole reason why we use normalized trace formulas in \stdl (see Section~\ref{sec:normalization}), and a primary reason why introduce the notion of recording the amount of time it takes for a hybrid program $\alpha$ to execute (see Section~\ref{sec:timing}). As a reminder here, for a hybrid program $\alpha$, executing $timed(\alpha,q)$ is equivalent to executing $\alpha$ while recording the amount of time the program takes to execute, following which the timed value is output as a fresh variable $q$. With that in mind, a trace of $\alpha; \beta$ satisfies $\phi \sqcap \square_{[a,b]} \psi$ if and only if 
\begin{itemize}
    \item for $a\leq q \leq b$, all traces of $\alpha$ satisfy $\square_{[a,q]}\psi$, and for traces of $\alpha$ that terminate at time $q$, all following traces of $\beta$ satisfy $\phi \sqand \square_{[0,b-q]}\psi$,
    \item for $a\leq b \leq q$, all traces of $\alpha$ satisfy $\square_{[a,b]}\psi$, and for traces of $\alpha$ that terminate at time $q$, all following traces of $\beta$ satisfy $\phi$,
    \item for $q \leq a \leq b$, for traces of $\alpha$ that terminate at time $q$, all following traces of $\beta$ satisfy $\phi \sqand \square_{[a-q,b-q]}\psi$.
\end{itemize}
These properties for the cases of the relative ordering of $a, b$ and $q$ are captured succinctly in rule ([;] $\sqcap_t$) using \texttt{min} and \texttt{max}. Rule ([;] $\sqcup_t$) is similar.

For the rules for non-deterministic finite repetition, let us first remember that as long as a trace $\alpha$ is finite, its finite repetition $\alpha^*$ will also be finite. In general, the rules attempt to reduce temporal properties of loops into either non-temporal properties of loops, or slightly more complex temporal properties on a program but without any loops. The idea here is to make the rules provable by ordinary, non-temporal induction. The key intuition behind rule ([*] $\sqcup_t$) comes from a very useful rule for repetition from d$\mathcal{L}$, which says that for a given trace formula $\pi$, the following is true:
$$\frac{[?true]\pi \land [\alpha^*;\alpha]\pi}{[\alpha^*]\pi} \text{ [;]}$$
Rule ([*] $\sqcap_t$) captures the fact that a trace of $\alpha^*$ satisfies $\phi\sqand\square_{[a,b]}\psi$ if and only if when $\alpha$ repeats zero times, $\phi$ is true, and if $a=0$ and $b=0$ then $\psi$ is true as well, or $\alpha^*$ runs first followed by $\alpha$, during which $\phi\sqand\square_{[a,b]}\psi$ with time interval shifting (similar to that for the sequential composition rules) holds. In rule ([*] $\sqcup_t$), the first disjunct expresses that $\dia_{[a,b]}\psi$ holds without repeating $\alpha$ if $a=0$ and $b\geq a$ and $\psi$ is true initially; the first conjunct of the second disjunct deals with the case where $\alpha$ repeats zero times and $\psi$ is false initially, while the second conjunct requires a sequential composition of $\alpha^*;\alpha$ to satisfy $\phi\sqor\dia_{[a,b]}\psi$ according to the rule ([;]) from d$\mathcal{L}$ mentioned above. Note that for rule ([*] $\sqcup_t$), the use of $\alpha^*;\alpha$ is equivalent to the use of $\alpha;\alpha^*$, and either variant of the sequential composition may be used. The rules (ind $\sqor_t$) and (con $\sqand_t$) extend the rules of induction (ind) and convergence (con) from d$\mathcal{L}$ to normalized trace formulas. Consistent with the rules from d$\mathcal{L}$, the rules (ind $\sqor_t$) and (con $\sqand_t$) are not equivalence relations (i.e., they do not have dual counterparts such that the negation of the premise and the conclusion is also a rule). The notation $\forall^\alpha$ from d$\mathcal{L}$ is a quantification over all variables that could be assigned by a hybrid program $\alpha$ in assignments or differential equations. Rule (ind $\sqor$) expresses that $\phi$ is inductive with exit clause $\dia_{[a,b]}\psi$ (i.e., $\phi$ is true after all traces $\sigma \in \hp{\alpha}$ where $\first\sigma \vDash \phi$, except when $\psi$ was true at some point in the interval $I$ during the execution of $\sigma$), while rule (con $\sqand$) shows that $\varphi$ is a variant of some trace $\sigma \in \hp{\alpha}$ (as in, its level $r$ decreases) during which $\psi$ is always true, and starting from an initial $r$, for an $r$ for which $\varphi(r)$ holds, it will ultimately be the case that $r\leq 0$ without $\psi$ being false if we repeat $\alpha^*$ often enough~\cite{CMU-CS-14-109}.

\subsubsection{Running Example: Traction Assist in Cars}\label{sec:casestudyproof} In this subsection, we present a proof sketch of the safety property for our example, highlighting how the property expressed in \stdl is reduced to an equivalent \dl formula to leverage the \dl calculus. For ease of understanding of the proof sketch, we only consider the second half of the hybrid program {\sf traction\_assist} (referred to as $\alpha$) -- though the application of the \stdl proof rules to the first half of the program is also fairly straightforward. We refer to the sequential composition components $\omega := -1$, $\varphi:=10$ and $\rho' = {\sf cruise(\omega)})$ in $\alpha$ as $\alpha_1$, $\alpha_2$, and $\alpha_3$ respectively. We then refer to the safety property $\phi$ as $\boxp{\alpha}(\psi_1 \sqor \dia_{[1,5]}\psi_2)$, where $\psi_1 \equiv \lnot{\sf no\_traction}$ and $\psi_2 \equiv (\rho < \rho_0)$.

For $\alpha \equiv (\alpha_1;(\alpha_2;\alpha_3))$, using the \stdl proof calculus, we get a proof tree of the form presented in Figure~\ref{fig:case_study_proofsketch}, where $\psi_3$ is obtained by applying rule $\boxp{'}\ \sqor_t$  with $a=\max(0,\max(0,1-q_1)-q_2)$ and $b=5-q_1-q_2$ as follows:
$$\psi_3 \equiv \boxp{\alpha_3}\left(\scalebox{0.93}{\begin{tabular}{c}
    $(b < a \land [\rho'=\omega \times k - \varphi \times j]\psi_1) \ \lor$\\
    $(([t := 0; \{x' = \theta, t' = 1\ \&\ (\chi \land t \leq a)\}; ?(t=a)]\psi_2)\ \land$ \\
    $\boxp{\rho'=\omega \times k - \varphi \times j \ \& \ (\lnot\psi_2)}\psi_1\ \land$\\
    $[t := 0; \{\rho'=\omega \times k - \varphi \times j, t' = 1\ \&\ (t \leq a)\}; ?(t=a)]$\\
    $\langle\{\rho'=\omega \times k - \varphi \times j, t'=1\ \&\ (t \leq b)\}\rangle\psi_2)$
\end{tabular}}\right)$$
The \stdl state formula $\psi_3$ can be proven further using solely the non-temporal rules from d$\mathcal{L}$.

\subsection{Soundness and Completeness of the \stdl Proof Calculus}
\begin{theorem}
The proof calculus for \stdl is sound. 
\end{theorem}

Since \stdl conservatively extends d$\mathcal{L}$, the soundness of the proof calculus of d$\mathcal{L}$ applies to \stdl as well. We present the proof of soundness for the rules introduced by the \stdl calculus.
\begin{proof}
We prove the soundness of individual rules. By induction on the proof trees, soundness of the entire proof system is a corollary.\\

$([:=]\ \sqand_t)$: For any state $v$, there is a unique terminating trace $\sigma \in \hp{x:=\theta}$ such that $\first\sigma=v$. From the trace semantics of hybrid programs, we know that $\sigma = (\hat{v},\hat{w})$ with $w=[x\mapsto val(v,\theta)]$. Therefore, $v \vDash [x:=\theta](\phi \sqand \square_{I}\psi)$ if and only if
\begin{itemize}
    \item for $I_s = 0$ and $I_e=0$, $w \vDash \phi$, $v \vDash \psi$, and $w \vDash \psi$, which is true if and only if $v \vDash \psi \land [x:=\theta](\phi \land \psi)$;
    \item for $I_s>0$ and $I_e \geq I_s$, $w \vDash \phi$, which is true if and only if $v \vDash [x:=\theta]\phi$.
\end{itemize}
In either case, it follows that $v \vDash [x:=\theta](\phi \sqand \square_{I}\psi)$ if and only if $v \vDash ((I_s=0 \land I_e=0)\land(\psi\land\boxp{x:=\theta}(\phi\land\psi))\lor((I_s>0 \land I_e\geq I_s)\land\boxp{x:=\theta}\phi)$.\\

$([:=]\ \sqor_t)$: Similar to the proof of soundness of $([:=]\ \sqand_t)$.\\

([?] $\sqand_t$): 
($\rightarrow$) Let $v \vDash ((a=0\land b=0)\land((\chi \wedge (\phi \lor \psi)) \vee (\neg \chi \wedge \psi)) \lor ((a>0\land b\geq a)\land(\lnot\chi \lor (\chi \land \phi)))$, and let $\sigma \in \hp{?\chi}$ with $\first\sigma = v$. If $v\vDash \lnot \chi$, then $\sigma = (\hat{v}, \hat{\Lambda})$ (i.e., $\sigma$ is the error trace). If $a=0$ and $b=0$, by our assumption, it follows that $v \vDash \psi$ (otherwise $v$ does not satisfy anything). Since $\sigma$ is a trace that occurs in zero time, it follows that $\sigma \vDash (\phi \sqand \square_{[a,b]}\psi)$. If, however, $v \vDash \chi$, then $\sigma = (\hat{v})$, and by our assumption, if $a>0$ and $b \geq a$, then $v\vDash \phi$ only (since the length of $\sigma$ is not long enough to determine the satisfiability of $\square_{[a,b]}\psi$). Therefore, $\sigma \vDash (\phi \sqand \square_{[a,b]}\psi)$ in this case as well.\\
($\leftarrow$) Conversely, assume that $v \vDash [?\chi](\phi \sqand \square_{[a,b]}\psi)$. Now, if $v \vDash \lnot \chi$, then $\sigma = (\hat{v}, \hat{\Lambda})$ (which is a non-terminating state), and $v\vDash \psi$ only when $a=0$ and $b=0$. Otherwise, $v\vDash \chi$ and $\sigma = (\hat{v})$, and therefore $v \vDash (\phi \land \psi)$ when $a=0$ and $b=0$, or $v \vDash \phi$ when $a>0$ and $b\geq a$. In either case, $v\vDash ((a=0\land b=0)\land((\chi \wedge (\phi \lor \psi)) \vee (\neg \chi \wedge \psi)) \lor ((a>0\land b\geq a)\land(\lnot\chi \lor (\chi \land \phi)))$.\\

([?] $\sqor_t$): Similar to the proof of soundness of [?] $\sqand_t$.\\

(['] $\sqand_t$): 
($\rightarrow$) Let $v \vDash (b < a \land [x'=\theta \ \& \ \chi]\phi) \lor (\lnot \chi \land (\lnot(a=0 \land b \geq a) \lor \psi)) \lor ([t:=0;\{x'=\theta, t'=1 \ \& \ (\chi \land t \leq a)\}; ?(t=a)] [\{x'=\theta, t'=1 \ \& \ (\chi \land t \leq b)\}]\psi \land [x'=\theta \ \& \ \chi]\phi)$, and let $\sigma \in \hp{x'=\theta \ \& \ \chi}$ such that $\first\sigma=v$. If $b<a$, it is only required that $\sigma \vDash \phi$ (since $\square_{[a,b]}\psi$ is trivially true in this case). If $v \vDash \lnot \chi$, then $\sigma$ is the non-terminating error trace $(\hat{v}, \hat{\Lambda})$ and $\sigma \vDash \square_{[a,b]}\psi$ (since $\psi$ is true when $a=0$ and $b\leq a$). Therefore, $\sigma\vDash (\phi \sqand \square_{[a,b]}\psi)$. If $v \vDash \chi$, however, then $\sigma = \{f\}$ for a real function $f$ defined on $D = [0, r]$ solution of $x' = \theta$, which satisfies $\chi$ on its domain of definition. Since $v \vDash [x'=\theta \ \& \ \chi]\phi$, for any $\sigma$ that terminates, $\sigma \vDash \phi$. For a $\sigma$ that does not terminate, $v \vDash [t:=0;\{x'=\theta, t'=1 \ \& \ (\chi \land t \leq a)\}; ?(t=a)] [\{x'=\theta, t'=1 \ \& \ (\chi \land t \leq b)\}]\psi$, and therefore $\sigma \vDash \square_{[a,b]}\psi$. In either case, $\sigma \vDash (\phi \sqand \square_{[a,b]}\psi)$.\\
($\leftarrow$) Conversely, assume $v \vDash [x'=\theta \ \& \ \chi](\phi \sqand \square_{[a,b]}\psi)$. By definition, there exists at least one trace $\sigma \in \hp{x'=\theta \ \& \ \chi}$ such that $\first\sigma = v$ and $\sigma \vDash \square_{[a,b]}\psi$. Now, if $v \vDash \lnot \chi$, then $v\vDash\psi$ if $a=0$ and $b\geq a$. Otherwise, for non-error traces of $x'=\theta \ \& \ \chi$,
\begin{itemize}
    \item for a terminating trace $\sigma \in \hp{x'=\theta \ \& \ \chi}$, we have that $\sigma \vDash \phi\sqand\square_{[a,b]}\psi$, and in particular, we have that $\sigma\vDash\phi$
    \item for any trace $\sigma \in \hp{x'=\theta \ \& \ \chi}$ (terminating or otherwise), since $\sigma \vDash \phi\sqand\square_{[a,b]}\psi$, in particular we have that $\sigma\vDash \square_{[a,b]}\psi$, and hence, $\sigma\vDash [t:=0;\{x'=\theta, t'=1 \ \& \ (\chi \land t \leq a)\}; ?(t=a)] [\{x'=\theta, t'=1 \ \& \ (\chi \land t \leq b)\}]\psi$.
\end{itemize}
Therefore, $v \vDash (b < a \land [x'=\theta \ \& \ \chi]\phi) \lor (\lnot \chi \land (\lnot(a=0 \land b \geq a) \lor \psi)) \lor ([t:=0;\{x'=\theta, t'=1 \ \& \ (\chi \land t \leq a)\}; ?(t=a)] [\{x'=\theta, t'=1 \ \& \ (\chi \land t \leq b)\}]\psi \land [x'=\theta \ \& \ \chi]\phi)$.\\

(['] $\sqor_t$): 
($\rightarrow$) Assume $v\vDash (b < a \land [x'=\theta \ \& \ \chi]\phi) \lor ((\chi \lor [t := 0; \{x' = \theta, t' = 1\ \&\ (\chi \land t \leq a)\}; ?(t=a)]\psi) \land \boxp{x' = \theta \ \& \ (\chi \land \lnot\psi)}\phi\land [t := 0; \{x' = \theta, t' = 1\ \&\ (\chi \land t \leq a)\}; ?(t=a)] \langle\{x' = \theta, t'=1\ \&\ (t \leq b)\}\rangle(\lnot \chi \lor \psi))$ and let $\sigma \in \hp{x'=\theta \ \& \ \chi}$ such that $\first\sigma=v$. If $b<a$, then $v\vDash [x'=\theta \ \& \ \chi]\phi$. If $v\vDash \lnot \chi$, then $\sigma$ is the non-terminating trace $(\hat{v}, \hat{\Lambda})$ such that $\sigma \vDash \dia_{[a,b]}\psi$. Therefore, $\sigma \vDash \phi \sqor \dia_{[a,b]}\psi$. If $v \vDash \chi$, however, then $\sigma = \{f\}$ for a real function $f$ defined on $D = [0, r]$ solution of $x' = \theta$, which satisfies $\chi$ on its domain of definition. If $\sigma \vDash \dia_{[a,b]}\psi$, then by definition, $\sigma \vDash \phi \sqor \dia_{[a,b]}\psi$. Otherwise, if $\sigma$ is terminating and no state of $\sigma$ satisfies $\psi$, we have that $\sigma \in \hp{x'=\theta \ \&\ (\chi \land \lnot \psi}$. From our assumption, we have that $\sigma \vDash \phi$, and as such, $\sigma \vDash \phi \sqor \dia_{[a,b]}\psi$. Lastly, for the case case where $\sigma \nvDash \dia_{[a,b]}\psi$, we cannot have a non-terminating $\sigma$. This is because such a $\sigma$ would verify $\chi \land \lnot \psi$ in all states, and could follow any trace $\sigma_\alpha \in \hp{x' = \theta}$, contradicting $v \vDash [t := 0; \{x' = \theta, t' = 1\ \&\ (\chi \land t \leq a)\}; ?(t=a)] \langle\{x' = \theta, t'=1\ \&\ (t \leq b)\}\rangle(\lnot \chi \lor \psi)$ in the process.\\
($\leftarrow$) Conversely, let $v\vDash [x'=\theta \ \& \ \chi](\phi \sqor \dia_{[a,b]}\psi)$, and let $\sigma \in \hp{x'=\theta \ \& \ \chi}$ such that $\first\sigma=v$. First, if $b<a$, the $\dia_{[a,b]}\psi$ is vacuously false, and since $\sigma \vDash (\phi \sqor \dia_{[a,b]}\psi)$, it must be the case that $\sigma\vDash \phi$. Otherwise, if $v \vDash \lnot \chi$, the only trace of $\hp{x'=\theta \ \& \ \chi}$ such that $\first\sigma = v$ is the trace $(\hat{v},\hat{\Lambda})$. Since this trace satisfies $\dia_{[a,b]}\psi$, we have that $v \vDash [t := 0; \{x' = \theta, t' = 1\ \&\ (\chi \land t \leq a)\}; ?(t=a)]\psi$. Therefore, in all cases, we have $v \vDash \chi \lor [t := 0; \{x' = \theta, t' = 1\ \&\ (\chi \land t \leq a)\}; ?(t=a)]\psi$. To prove that $v\vDash \boxp{x' = \theta \ \& \ (\chi \land \lnot\psi)}\phi$, we need only consider terminating tracing. Let $\sigma$ be a terminating trace of $\hp{x'=\theta \ \&\ (\chi \land \lnot \psi)}$. Then, in particular, $\sigma \in \hp{x'=\theta \ \&\ \chi}$, and as such, $\sigma \vDash \phi \sqor \dia_{[a,b]}\psi$. Since $\sigma$ also has $\lnot \psi$ as domain constraint, it follows that $\sigma \nvDash \dia_{[a,b]}\psi$, and as such, $\sigma \vDash \phi$. Finally, to prove the third conjunct of the rule, let us first consider the case where $v\vDash [t := 0; \{x' = \theta, t' = 1\ \&\ (\chi \land t \leq a)\}; ?(t=a)] \langle\{x' = \theta, t'=1\ \&\ (t \leq b)\}\rangle\lnot \chi$. In this case, there is no non-terminating trace $\sigma \in \hp{x'=\theta \ \&\ \chi}$ such that $\first\sigma = v$. For the case where $v\vDash [t := 0; \{x' = \theta, t' = 1\ \&\ (\chi \land t \leq a)\}; ?(t=a)] \langle\{x' = \theta, t'=1\ \&\ (t \leq b)\}\rangle \psi$, there exists a unique non-terminating trace $\sigma \in \hp{x'=\theta \ \&\ \chi}$ such that $\first\sigma = v$. By our assumption, we have that $\sigma \vDash \dia_{[a,b]}\psi$. This means that $\psi$ has to be true in some state that is reached by trace $\sigma$, and this notion is logically equivalent to $[t := 0; \{x' = \theta, t' = 1\ \&\ (\chi \land t \leq a)\}; ?(t=a)] \langle\{x' = \theta, t'=1\ \&\ (t \leq b)\}\rangle \psi$. From both of the cases mentioned above, we get $v\vDash [t := 0; \{x' = \theta, t' = 1\ \&\ (\chi \land t \leq a)\}; ?(t=a)] \langle\{x' = \theta, t'=1\ \&\ (t \leq b)\}\rangle(\lnot \chi \lor \psi)$.\\

$([\ ]\leadsto)$, $(\diap{}\leadsto)$: Soundness of rules $([\ ]\leadsto)$ and $(\diap{}\leadsto)$ is a corollary of Lemma~\ref{lemma:soundnessnormalization}.\\

$([\cup]\ \xi)$: For any state $v$, we have $v \vDash [\alpha]\xi \land [\beta]\xi$ if and only if for all traces $\sigma_\alpha \in \hp{\alpha}$ such that $\first \sigma_\alpha = v$, it follows that $\sigma_\alpha \vDash \xi$, and for all traces $\sigma_\beta \in \hp{\beta}$ such that $\first \sigma_\beta = v$, it follows that $\sigma_\beta \vDash \xi$. This is true if and only if for all traces $\sigma \in \hp{\alpha \cup \beta}$ such that $\first \sigma = v$, it follows that $\sigma \vDash \xi$, which in turn is true if and only if $v \vDash [\alpha \cup \beta]\xi$.\\  

Before we can prove soundness for the rules for sequential composition, we need the following lemma describing the relationship between $q$ and $|\sigma|$ for $\sigma \in \hp{timed(\alpha,q)}$:
\begin{lemma}[Timing and the lengths of traces of hybrid programs]\label{lemma:timingandsigmalength}
Given a hybrid program $\alpha$ and a trace $\sigma \in \hp{timed(\alpha,q)}$ that the execution of $timed(\alpha,q)$ follows, the time taken to execute the program is equal to the length of $\sigma$ (i.e., $q = |\sigma|$).
\end{lemma}
\begin{proof}
Lemma~\ref{lemma:timingandsigmalength} is a direct consequent of Definition~\ref{def:tracelength} and Definition~\ref{def:timing}, keeping in mind the semantics of hybrid programs from Definition~\ref{def:hybridprograms}.
\end{proof}

We can now continue with the proof of soundness of the \stdl calculus.\\

([;] $\sqand_t$): 
($\rightarrow$) Assume that for some state $v$, it is true that $v \vDash [timed(\alpha,q)]([\beta](\phi \sqcap \square_{[\max(0,a-q),b-q]} \psi)\sqcap \square_{[a,\min(b,q)]} \psi)$, and let $\sigma|_{\var - \{q\}} \in \hp{\alpha;\beta}$ such that $v=\first\sigma$. If $\sigma|_{\var - \{q\}} \in \hp{\alpha}$ is an infinite trace, then by Lemma~\ref{lemma:timedprograms}, $\sigma|_{\var - \{q\}} \in \hp{timed(\alpha,q)}$ and by the assumption, $\sigma|_{\var - \{q\}} \vDash \square_{[a,\min(b,+\infty)]} \psi \equiv \square_{[a,b]}\psi$. Otherwise, there exists a terminating trace $\sigma_\alpha \in \hp{\alpha}$ such that $|\sigma_\alpha| = q$ (by Lemma~\ref{lemma:timingandsigmalength}) and a trace (infinite or otherwise) $\sigma_\beta \in \hp{\beta}$ such that $\sigma = \sigma_\alpha \circ \sigma_\beta$. By the assumption, $\sigma_\alpha|_{\var - \{q\}} \vDash \square_{[a,\min(b,q)]} \psi$ and $\sigma_\beta \vDash (\phi \sqcap \square_{[\max(0,a-q),b-q]} \psi$. Now, depending on the value of $q$, there are three possible orders of $a$, $b$, and $q$: $a\leq q \leq b$, $q\leq a \leq b$, and $a\leq b \leq q$. Keeping in mind the fact that for $b<a$, $\square_{[a,b]}\psi$ is vacuously true while $\dia_{[a,b]}\psi$ is vacuously false, we can see that
\begin{itemize}
    \item for $a\leq q \leq b$, $\sigma_\alpha \vDash \square_{[a,q]} \psi$ and $\sigma_\beta \vDash (\phi \sqcap \square_{[0,b-q]} \psi)$; 
    \item for $q\leq a \leq b$, $\sigma_\beta \vDash (\phi \sqcap \square_{[a-q,b-q]} \psi)$;
    \item for $a \leq b \leq q$, $\sigma_\alpha \vDash \square_{[a,b]} \psi$ and $\sigma_\beta\vDash\phi$.
\end{itemize}
In all cases, $\sigma_\alpha|_{\var - \{q\}} \circ \sigma_\beta \vDash (\phi \sqand \square_{[a,b]}\psi)$. By Lemma~\ref{lemma:timedprograms}, we get $\sigma_\alpha|_{\var - \{q\}} \circ \sigma_\beta \equiv \sigma$. Therefore, $\sigma \vDash (\phi \sqand \square_{[a,b]}\psi)$.\\
($\leftarrow$) Conversely, let $v \vDash [\alpha;\beta](\phi \sqand \square_{[a,b]}\psi)$. Let $\sigma_\alpha \in \hp{\alpha}$ such that $v = \first\sigma_\alpha$. If $\sigma_\alpha$ is infinite, then $\sigma_\alpha \in \hp{\alpha; \beta}$, and as such, $\sigma_\alpha \vDash \square_{[a,b]}\psi$. Otherwise, let $\sigma_\beta \in \hp{\beta}$ such that $\sigma_\alpha \circ \sigma_\beta$ is well-defined. Again, since $|\alpha| = q$ (by Lemma~\ref{lemma:timingandsigmalength}), there are three possible orders of $a$, $b$, and $q$. It is easy to see that for any relative ordering of $a$, $b$, and $q$, $\sigma_\alpha \vDash \square_{[a,\min(b,q)]}\psi$, and $\sigma_\beta\vDash(\phi \sqand \square_{[\max(0,a-q),b-q]}\psi)$. There is a universal quantifier on $\sigma_\beta$, so $\sigma_\alpha \vDash [\beta](\phi \sqcap \square_{[\max(0,a-q),b-q]} \psi)$. Keeping in mind that the choice of $\sigma_\alpha$ was arbitrary, and by using Lemma~\ref{lemma:timedprograms}, it follows that $v \vDash [timed(\alpha,q)]([\beta](\phi \sqcap \square_{[\max(0,a-q),b-q]} \psi)\sqcap \square_{[a,\min(b,q)]} \psi)$.\\

([;] $\sqor_t$): Similar to the proof of soundness of [;] $\sqand_t$.\\

([*] $\sqand_t$): 
($\rightarrow$) Assume $v\vDash (\phi\land(\lnot(a=0 \land b=0) \lor \psi)) \wedge [timed(\alpha^*,q)][\alpha](\phi \sqcap \square_{[max(0,a-q),b-q]} \psi)$ and let $\sigma \in \hp{\alpha^*}$. If $\sigma \in \hp{\alpha^0} = \hp{?\true}$, then $\sigma = (\hat{v})$. Since we have $v \vDash \phi$ and $v\vDash [timed(\alpha^0,q)][\alpha]\phi \sqcap \square_{[a,b]} \psi$ (note that for $timed(\alpha^0,q)$, we have $q=0$, and such a trace satisfies $\psi$ only if $a=0$ and $b=0$), it follows that $\sigma \vDash (\phi \sqand \square_{[a,b]}\psi)$. Otherwise, there exits $n \geq 1$ such that $\sigma = \sigma_1 \circ \dots \circ \sigma_n$, where $\sigma_i \in \hp{\alpha}$ for any $i \in \{1, \dots, n\}$. Then, $\sigma_1 \circ \dots \circ \sigma_{n-1} \in \hp{\alpha^*}$ and $\sigma_n \in \hp{\alpha}$. By our assumption, $\sigma_n \vDash \square_{[max(0,a-q),b-q]} \psi$, and as such, $\sigma \vDash \square_{[a,b]}\psi$. Furthermore, if $\sigma$ terminates, then $\sigma_n$ terminates, and as such, we get $\sigma_n \vDash \phi$ from our assumption. Therefore, we have $\sigma \vDash \phi \sqand \square_{[a,b]}\psi$.\\
($\leftarrow$) Conversely, if $v\vDash [\alpha^*](\phi \sqcap \square_{[a,b]} \psi) $, then in particular, we have $(\hat{v}) \vDash (\phi \sqcap \square_{[a,b]} \psi)$. This implies that $v \vDash \phi$. Now, let $\sigma_\alpha \in \hp{\alpha^*}$ such that $\first\sigma_\alpha=v$ and $\sigma_\alpha$ terminates, and let $\sigma_\beta \in \hp{\alpha}$ such that $\first\sigma_\beta = \last\sigma_\alpha$. Then, $\sigma_\alpha \circ \sigma_\beta \in \hp{\alpha^*}$. Therefore, by our assumption, we get $\sigma_\alpha \circ \sigma_\beta \vDash \phi \sqcap \square_{[a,b]} \psi$. In particular, for a value of $q$ returned from $timed(\alpha^*,q)$ for $\sigma_\alpha$, we get $\sigma_\beta \vDash \phi \sqcap \square_{[max(0,a-q),b-q]} \psi$. Therefore, it follows that $v\vDash (\phi\land(\lnot(a=0 \land b=0) \lor \psi)) \wedge [timed(\alpha^*,q)][\alpha](\phi \sqcap \square_{[max(0,a-q),b-q]} \psi)$.\\

([$^{*n}$] $\sqor_t$): From the semantics of hybrid programs, we know that $\hp{\alpha^*} = \hp{?\true \cup \alpha^*;\alpha}$. Therefore, $v\vDash \hp{\alpha^*}(\phi \sqor \dia_{[a,b]}\psi)$ if and only if $v\vDash \hp{?\true \cup \alpha^*;\alpha}(\phi \sqor \dia_{[a,b]}\psi)$. By rule $([\cup]\ \xi)$, this is true if and only if $v\vDash \hp{?\true}(\phi \sqor \dia_{[a,b]}\psi)$ and $v\vDash \hp{\alpha^*;\alpha}(\phi \sqor \dia_{[a,b]}\psi)$. By rule $([?]\ \sqcup_t)$, $v\vDash \hp{?\true}(\phi \sqor \dia_{[a,b]}\psi)$ is itself equivalent to $v \vDash ((a=0\land b=0)\land (\phi \lor \psi)) \lor ((a>0\land b\geq a)\land \phi)$. Therefore, $v \vDash \hp{\alpha^*}(\phi \sqor \dia_{[a,b]}\psi)$ if and only if $v \vDash ((a=0\land b=0)\land (\phi \lor \psi)) \lor ((a>0\land b\geq a)\land \phi) \land \hp{\alpha^*;\alpha}(\phi \sqor \dia_{[a,b]}\psi)$. However, $v\vDash ((a=0\land b\geq a)\land \psi)$ implies $v\vDash \hp{\alpha^*;\alpha}(\phi \sqor \dia_{[a,b]}\psi)$. As such, we get $v\vDash \hp{\alpha^*}(\phi \sqor \dia_{[a,b]}\psi)$ if and only if $v\vDash (\psi \land (a=0\land b\geq a)) \lor (\phi \wedge [\alpha^*; \alpha](\phi \sqcup \lozenge_{[a,b]} \psi)$.\\

(ind $\sqcup_t$): Assume $v\vDash (\phi \implies [\alpha](\phi \sqcup \lozenge_{[a,b]} \psi))$ and $v\vDash \phi$. Let $\sigma \in \hp{\alpha^*}$. The proof is trivial for the case where $\sigma = (\hat{v})$. For any other $\sigma$, there exists $n \geq 1$ such that $\sigma  = \sigma_1 \circ \dots \circ \sigma_n$. If there exists a $\sigma_i$ such that $\sigma_i \vDash \dia_{[a,b]}\psi$, as is always the case where $\sigma$ is non-terminating, we get that $\sigma \vDash \dia_{[a,b]}\psi$. Otherwise, for any $i \in \{1, \dots, n\}$, since $\sigma_i \in \hp{\alpha}$, instantiating the premise using the universal $\forall^\alpha$ (this is necessary since the premise may behave differently for different states otherwise), if $\first\sigma_i\vDash \phi$, we get $\sigma_i \vDash \phi \sqor \dia_{[a,b]}\psi$. However, since $\sigma_i \nvDash \dia_{[a,b]}\psi$, we have $\last \sigma_i \vDash \phi$. Since $v\vDash \phi$, by induction on $i$, we get $\sigma \vDash \phi$, which leads to the conclusion of the rule.\\

(con $\sqcap_t$): Assume $v \vDash \forall^\alpha \forall r > 0 \ (\varphi(r) \implies \diap{\alpha}(\varphi(r-1)\sqand \square_{[a,b]} \psi))$ and $v\vDash (\exists r. \varphi(r)) \wedge \psi$. Then, there exists a $d \in \mathbb{R}$ such that $v \vDash \varphi(d)$. We prove the rule using well-founded induction on $d$. If $d \leq 0$, we have $(\hat{v}) \vDash ((\exists r. r \leq 0 \land \varphi(r)) \sqand \square_{[a,b]} \psi)$, where $(\hat{v})\in\hp{\alpha^*}$ for the case where $\alpha$ repeats zero times. If, however, $d>0$, we know that $v \vDash \varphi(d)$ and $v\vDash \varphi(d) \implies \diap{\alpha}(\varphi(d-1)\sqand \square_{[a,b]} \psi)$. Therefore, there exists an trace $\sigma_1 \in \hp{\alpha}$ such that $\sigma_1 \vDash(\varphi(d-1)\sqand \square_{[a,b]} \psi) $. Since $\last\sigma_1 \vDash \varphi(d-1)$, if $d-1 \leq 0$, we are done with the proof; otherwise, we can construct a similar $\sigma_2$ such that $\sigma_2 \vDash (\varphi(d-2)\sqand \square_{[a,b]} \psi)$. We can continue until $d \leq 0$, and this induction is well-founded because $d$ decrease by 1 for each step. We have thus constructed $\sigma = \sigma_1 \circ \dots \circ \sigma_n \in \hp{\alpha^*}$ such that each $\sigma_i \vDash \square_{[a,b]}\psi$ -- and thus $\sigma \vDash\square_{[a,b]}\psi$ -- and $\last\sigma_n = \last\sigma \vDash (\exists r. r \leq 0 \land \varphi(r))$. Therefore, we have $\sigma \vDash (\exists r. r \leq 0 \land \varphi(r)) \sqand \square_{[a,b]} \psi$.
\end{proof}

\begin{theorem}
\stdl is non-axiomatizable.
\end{theorem}
\begin{proof}
Discrete and continuous fragments of d$\mathcal{L}$ were proved to not be axiomatizable in \cite{platzer2008differential,platzer2012logics}. Since \stdl extends d$\mathcal{L}$, discrete and continuous fragments of \stdl are also non-axiomatizable. Therefore, in general, \stdl is non-axiomatizable.
\end{proof}

Even though \stdl is non-axiomatizable in general, its proof system restricted programs without repetitions is complete relative to first-order logic of differential equations (i.e., first-order real arithmetic augmented with formulas expressing properties of differential equations)~\cite{platzer2008differential,platzer2012logics}, as was shown to be the case for d$\mathcal{L}$.
\begin{theorem}
The proof calculus for \stdl restricted to programs without non-deterministic finite repetitions is complete relative to first-order logic of differential equations.
\end{theorem}
\begin{proof}
If we restrict \stdl to programs without repetition, the proof calculus for \stdl reduces temporal properties to non-temporal properties to leverage the calculus of d$\mathcal{L}$, which is proven to be complete relative to first-order logic of differential equations~\cite{platzer2008differential,platzer2012logics,Platzer2010}. More specifically, any temporal rule in the \stdl calculus transforms a normalized trace formula to a simpler normalized trace formula either without a temporal operator or with a temporal operator following a simpler, decomposed program. Every proof rule is an equivalence relation (i.e., the premise is equivalent to the conclusion), and Lemma~\ref{lemma:normalizedexistence} ensures that every trace formula in the syntax of \stdl can be converted into a normalized trace formula able to be handled by the \stdl calculus. Therefore, the relative completeness result of d$\mathcal{L}$ extends to \stdl limited to programs without repetition.
\end{proof}
Indeed, we conjecture that the \stdl proof calculus is complete relative to first-order logic of differential equations for all \stdl programs. We leave a formal proof of full relative completeness as future work.
\section{Future Work}\label{sec:future}
We plan on working on the following improvements to \stdl as future work:
\begin{itemize}
    \item \textit{Proving full relative completeness of \stdl:} While we prove that the calculus presented in \stdl restricted to programs without non-deterministic repetition is complete relative to first-order logic of differential equations, we conjecture that the calculus is indeed complete relative to first-order logic of differential equations for all programs. We have yet to prove this conjecture formally.
    \item \textit{Allowing for nested temporal operators in \stdl:} The fragment of STL currently supported by our work does not include properties with nested temporal operators, such as $\square_{[a,b]}\dia_{[c,d]}\phi$, to simplify the proof system. We do not consider this to be a significant drawback, since the fragment of STL considered is sufficient to cover a large amount of properties of interest expressed in previous case studies involving STL~\cite{boundedstl2019,raman2014model,jha2019telex,donze2012temporal,raman2015reactive}. Nevertheless, we hope to remove this restriction in the future to further increase the expressive power of \stdl.
    \item \textit{Implementing the rules for \stdl}: We hope to implement the rules for the \stdl proof system into a theorem prover for hybrid systems such as KeYmaera~\cite{platzer2008keymaera,fulton2015keymaera}.
\end{itemize}
\section{Related Work}\label{sec:relatedwork}
In this section, we explore works related to reasoning about properties of hybrid systems and using STL for monitoring and verfication purposes. 

STL~\cite{maler2004monitoring,maler2013monitoring} was introduced for monitoring properties over continuous signals, and has since been studied widely, e.g., in Deshmukh et al.~\cite{deshmukh2017robust}, Donz{\'e} and Maler~\cite{donze2010robust}, Maler et al.~\cite{maler2008checking}. Most uses of STL have been mainly for monitoring purposes. However, there has been some work done on studying temporal properties of hybrid systems in the context of model checking. Mysore et al.~\cite{mysore2005algorithmic} examine model checking of semi-algebraic hybrid systems for Timed Computation Tree Logic properties. Their work focuses on bounded model checking for differential equations with polynomial solutions only, while we allow for more general polynomial differential equations. Roehm et al.~\cite{roehm2016stl} define a new reachset temporal logic (RTL) and transform STL properties to RTL properties to perform model checking of continuous and hybrid systems. More recently, Bae and Lee~\cite{boundedstl2019} explore a bounded model checking of signal temporal logic properties using syntactic separation of STL. For both \cite{boundedstl2019} and \cite{mysore2005algorithmic}, the applications presented focus on bounded safety verification, while our work allows unbounded safety verification. Better still, our proof system enables proving strong liveness properties for hybrid systems, a trait not present in works like \cite{boundedstl2019}, \cite{roehm2016stl}, and \cite{mysore2005algorithmic}. 

Process logic~\cite{harel1982process,nishimura1980descriptively,pratt1979process} originally used Pnueli's temporal logic~\cite{pnueli1977temporal} in the context of Harel et al.'s dynamic logic~\cite{harel2001dynamic} for temporal reasoning of hybrid systems. However, it is restricted to discrete programs and only considers an abstract notion of atomic programs, without supporting explicit assignments and tests. Platzer~\cite{platzer2008differential,platzer2012logics,platzer2010logical} introduce differential dynamic logic (d$\mathcal{L}$) to reason about the end states of a hybrid program, later followed by differential temporal dynamic logic (dTL)~\cite{Platzer2010} to reason about intermediate states of hybrid programs throughout the execution of the program using some temporal operators of linear temporal logic. Jeannin and Platzer~\cite{jeannin2014dtl} present dTL\textsuperscript{2}, a logic that extends dTL and allows for alternating program and temporal modalities. While our work draws on the technical machinery from dTL\textsuperscript{2}, the logic has a significant drawback compared to \stdl in that it does not support reasoning about properties in given time intervals. This nature of reasoning not only is often crucial to proving safety of hybrid systems, but also allows for expressing a significantly richer set of liveness properties. 

Sogokon et al. present a proof method for proving eventuality properties~\cite{sogokon2015direct} and persistence properties~\cite{sogokon2017verifying} in hybrid systems. Their methods focus on properties of the form $\dia_{[0,t]}\square_{[0,\infty)}P$, whereas our formalism is more general but does not support alternating temporal modalities -- the properties that the two results focus on are complementary to each other.
Note, however, that their formalism operates on the level of hybrid automata~\cite{alur1992hybrid,henzinger2000theory}, which unlike hybrid programs, do not enjoy the property of having a compositional semantics that can be used to verify systems by verifying properties of their parts in a theorem prover.
Tan and Platzer~\cite{tan2020axiomatic} present an axiomatic approach for deductive verification of existence and liveness for ordinary differential equations with \dl, but their approach only focuses on liveness for differential equations, and not entire hybrid systems. They also only work on formulas of the form $\langle\alpha\rangle P$, which is a fairly limited form of liveness.

Zhou et al.~\cite{chaochen1992extended} present a duration calculus for hybrid real-time systems extended by mathematical expressions with derivatives of state variables. The system that they present requires external mathematical reasoning about continuity and derivatives. Davoren and Nerode~\cite{davoren2000logics} study hybrid systems in the context of the propositional $\mu$-calculus. They provide a calculus to prove formulas in their systems, but with a propositional system (and not a first-order one). Furthermore, they do not provide specific rules in their proof system to handle ordinary differential equations.
\section{Conclusion}\label{sec:conclusion}
In this work, we introduce signal temporal dynamic logic (\stdl), a logic that extends and combines differential dynamic logic (d$\mathcal{L}$) with a fragment of signal temporal logic (STL). \stdl is a conservative extension of d$\mathcal{L}$ and allows reasoning not only about the final states of a hybrid system, but also the intermediate states of a hybrid system in given time intervals. While STL was originally intended to be a logic for monitoring systems, and has widely been used for exactly that purpose, we show that STL can very well be used for deductive verification of hybrid systems. \stdl allows us to prove a greater set of both safety and liveness properties than was possible with logics preceding \stdl. We provide a semantics and a sound proof calculus for \stdl, along with proofs of soundness and relative completeness.

\subsection*{Acknowledgements}
The authors would like to thank Nikos Ar\'echiga for insightful discussions.
Toyota Research Institute (``TRI'') provided funds to assist the authors with their research, but this article solely reflects the opinions and conclusions of its authors and not TRI or any other Toyota entity.

\bibliographystyle{ACM-Reference-Format}
\bibliography{references}










\end{document}